\documentclass[11pt]{article}

\usepackage[utf8]{inputenc}
\usepackage{graphicx}
\graphicspath{ {./images/} }
\usepackage{amssymb}
\usepackage{enumitem}
\usepackage{bbm}
\usepackage[margin=1in]{geometry}
\usepackage{amsmath}
\usepackage[ruled,vlined]{algorithm2e}
\usepackage{subfig}
\usepackage{mathtools}
\usepackage{scalerel}
\usepackage{pst-node}
\usepackage{tikz-cd} 
\usepackage{comment}
\definecolor{Black}{rgb}{0,0,0}
\usepackage[linktocpage=true,
pagebackref=true,colorlinks,
linkcolor=Black,citecolor=Black,
bookmarks,bookmarksopen,bookmarksnumbered]{hyperref}

\usepackage{cleveref}
\usepackage{amsthm}
\usepackage{thmtools,blindtext}
\usepackage{thm-restate}
\usepackage{xcolor}
\usepackage[nottoc]{tocbibind} 

\newtheorem{theorem}{Theorem}
\newtheorem{corollary}[theorem]{Corollary}
\newtheorem{lemma}[theorem]{Lemma}

\newtheorem{definition}[theorem]{Definition}

\newtheorem{question}[theorem]{Question}

\DeclarePairedDelimiter{\ceil}{\lceil}{\rceil}

\newcommand{\FF}{\mathbbm{F}}
\newcommand{\ZZ}{\mathbbm{Z}}

\usepackage{stmaryrd}

\newcommand{\sem}[1]{\llbracket#1\rrbracket}
\newcommand{\uhr}{\upharpoonright}
\newcommand{\AND}{\textsc{and}}
\newcommand{\OR}{\textsc{or}}
\newcommand{\NOT}{\textsc{not}}
\newcommand{\MAJ}{\textsc{maj}}
\newcommand{\Word}{{\normalfont\textsc{Word}}}
\newcommand{\Stab}{{\normalfont\mathrm{Stab}}}
\newcommand{\Diag}{{\normalfont\mathrm{Diag}}}
\newcommand{\AC}{\mathsf{AC}}
\newcommand{\TC}{\mathsf{TC}}
\newcommand{\BAR}[1]{\overline{#1}{}}

\title{Symmetric Formulas for Products of Permutations}
\author{William He\\ Duke University \and Benjamin Rossman\\Duke University}

\begin{document}

\maketitle

\begin{abstract}
We study the formula complexity of the word problem $\Word_{S_n,k} : \{0,1\}^{kn^2} \to \{0,1\}$: given $n$-by-$n$ permutation matrices $M_1,\dots,M_k$, compute the $(1,1)$-entry of the matrix product $M_1\cdots M_k$. 
An important feature of this function is that it is invariant under action of $S_n^{k-1}$ given by
\[
  (\pi_1,\dots,\pi_{k-1})(M_1,\dots,M_k) = (M_1\pi_1^{-1},\pi_1M_2\pi_2^{-1},\dots,\pi_{k-2}M_{k-1}\pi_{k-1}^{-1},\pi_{k-1}M_k).
\]
This symmetry is also exhibited in the smallest known unbounded fan-in $\{\AND,\OR,\NOT\}$-formulas for $\Word_{S_n,k}$, which have size $n^{O(\log k)}$.

In this paper we prove a matching $n^{\Omega(\log k)}$ lower bound for $S_n^{k-1}$-invariant formulas  computing $\Word_{S_n,k}$. This result is motivated by the fact that a similar lower bound for unrestricted (non-invariant) formulas would separate complexity classes $\mathsf{NC}^1$ and $\mathsf{Logspace}$.

Our more general main theorem gives a nearly tight $n^{d(k^{1/d}-1)}$ lower bound on the $G^{k-1}$-invariant depth-$d$ $\{\MAJ,\AND,\OR,\NOT\}$-formula size of 
$\Word_{G,k}$ for any finite simple group $G$ whose minimum permutation representation has degree~$n$. We also give nearly tight lower bounds on the $G^{k-1}$-invariant depth-$d$ $\{\AND,\OR,\NOT\}$-formula size in the case where $G$ is an abelian group.
\end{abstract}

\section{Introduction}

\subsection{$P$-invariant complexity}

Let $P$ be a permutation group on $m$ elements (i.e.,\ a subgroup of the symmetric group $S_m$). There is a natural action of $P$ on the set of Boolean function $f : \{0,1\}^m \to \{0,1\}$ given by $(\pi f)(x_1,\dots,x_m) = f(x_{\pi(1)},\dots,x_{\pi(m)})$. We say that $f$ is {\em $P$-invariant} if $\pi f = f$ for all $\pi \in P$ (i.e.,\ $f$ is symmetric under the action of $P$). Many important Boolean functions are invariant under group actions on coordinates $1,\dots,m$: {\em symmetric functions} such as \textsc{Parity} and \textsc{Majority} are fully $S_m$-invariant, while {\em graph properties} such as \textsc{Connectivity} and \textsc{3-Colorability} are invariant under $P = S_n$ acting on $m = \binom{n}{2}$ edge-indicator variables.

Any permutation group $P$ also acts on combinatorial devices that compute functions on $\{0,1\}^m$, such as $m$-variable circuits, formulas, branching programs, etc.  In this paper, we focus on {\em formulas} in both the $\mathsf{AC}$ basis (unbounded fan-in $\AND$, $\OR$ gates) and the $\mathsf{TC}$ basis (unbounded fan-in $\MAJ$ gates) with negations on input literals.  Here we view formulas as rooted trees in which leaves (``inputs'') are labeled by literals $\texttt x_1,\overline{\texttt x}_1,\dots,\texttt x_m,\overline{\texttt x}_m$ or constants $\texttt 0,\texttt 1$, and non-leaves (''gates'') are labeled by gate types. We treat the children of a gate as an unordered multiset, so that two formulas are identical if and only they are isomorphic as labeled rooted trees.

$P$ acts on the set of formulas by relabeling literals $\texttt x_i$ to $\texttt x_{\pi(i)}$ and $\overline{\texttt x}_i$ to $\overline{\texttt x}_{\pi(i)}$.
A formula $\Phi$ is said to be {\em $P$-invariant} if $\pi \Phi = \Phi$ for all $\pi \in P$.
Every $P$-invariant formula computes a $P$-invariant Boolean function. On the other hand, a $P$-invariant Boolean function $f$ may be computed (often more efficiently) by a non-$P$-invariant formula $\Phi$; in this case, we would say that $\Phi$ is {\em semantically $P$-invariant}, but not {\em syntactically $P$-invariant}.

The action of $P$ on formulas preserves parameters such as {\em depth} (the maximum number of gates on a leaf-to-root branch) and {\em size} (the number of leaves labeled by literals). 
For general Boolean functions $f$, we write $\mathcal L_{\mathsf{AC}_d}(f)$ and $\mathcal L_{\mathsf{TC}_d}(f)$ for the minimum size of a depth-$d$ $\mathsf{AC}$ (respectively $\mathsf{TC}$) formulas that computes $f$, and we write $\mathcal L_{\mathsf{AC}}(f)$ and $\mathcal L_{\mathsf{TC}}(f)$ allowing unrestricted depth.
When $f$ is $P$-invariant, we may consider the corresponding {\em $P$-invariant complexity measures} (indicated by superscript $P$):
\[
\begin{array}{cccccc}
  \mathcal L_{\mathsf{TC}_d}(f)
  &\le&
  \mathcal L_{\mathsf{AC}_d}(f)\\
  \rotatebox[origin=c]{-90}{$\le$}\vphantom{\Big|}
  &&\rotatebox[origin=c]{-90}{$\le$}\vphantom{\Big|}\\  
  \mathcal L^P_{\mathsf{TC}_d}(f)
  &\le&
  \mathcal L^P_{\mathsf{AC}_d}(f)
\end{array}
\qquad\qquad
\begin{array}{ccccc}
  \mathcal L_{\mathsf{TC}}(f)
  &\le&
  \mathcal L_{\mathsf{AC}}(f)
  &=&
  \mathcal L_{\mathsf{DeMorgan}}(f)\\
  \rotatebox[origin=c]{-90}{$\le$}\vphantom{\Big|}&
  &\rotatebox[origin=c]{-90}{$\le$}\vphantom{\Big|}\\
  \mathcal L^P_{\mathsf{TC}}(f)
  &\le&
  \mathcal L^P_{\mathsf{AC}}(f)
\end{array}
\]

Invariant circuit complexity has been previously studied in the context of descriptive complexity, an area concerned with characterizing complexity classes in terms of definability in different logics. Here the languages/Boolean functions considered are graph properties (or in general isomorphism-invariant properties of finite relational structures), and $P$ is the symmetric group $S_n$ acting on $m = \binom{n}{2}$ edge-indicator variables (or $m = \sum_{i=1}^t n^{r_i}$ indicator variables for properties of structures with relations of arity $r_1,\dots,r_t$).
Denenberg, Gurevich and Shelah \cite{Denenberg1986} showed that $P$-invariant $\mathsf{AC}$ circuits of polynomial size and constant depth (subject to a certain uniformity condition) capture precisely the first-order definable graph properties.
More recently, Anderson and Dawar \cite{Anderson2016} established a correspondence between polynomial-size $P$-invariant $\mathsf{TC}$ circuits and definability in fixed-point logic with counting. 
For both classes of formulas, {\em exponential gaps} are known between the $P$-invariant complexity and non-invariant complexity of explicit graph properties.

\subsection{The word problem over a finite permutation group}

Let $G \le S_n$ be a finite permutation group, which we assume to be transitive. Specifically, we will be interested in the case $G = S_n$, as well as the case where $G$ is a simple group and $G \le S_n$ is a faithful permutation representation of minimum degree (e.g.,\ $G$ is the alternating group $A_n$, or $G$ is a cyclic group of prime order $p$ and $n=p$).

We write $\BAR{G}$ for the set of permutation matrices corresponding to elements of $G$. For $k \in \mathbbm N$, we view $\BAR{G}^k$ as a subset of Hamming cube $\{0,1\}^{kn^2}$.
We write elements of $\BAR{G}^k$ as $k$-tuples of matrices $(M_1,\dots,M_k)$, and we identify coordinates of $\{0,1\}^{kn^2}$ with Boolean variables $M_{i,a,b}$ for $i \in [k]$ and $a,b \in [n]$.

\begin{definition}\label{df:word problem}
The {\em length-$k$ word problem} for $G$ is the Boolean function $\Word_{G,k} : \BAR{G}^k \to \{0,1\}$ that outputs the $(1,1)$-entry of the matrix product $M_1 \cdots M_k$. 
\end{definition}

This is one natural version of the ``word problem'' for a finite permutation group. This should not be confused with the word problem for finitely presented groups (see Section \ref{sec:fp groups}).  Note that the problem of determining whether $M_1\cdots M_k$ is the identity reduces to $n$ instances of $\Word_{G,k}$.

We are interested in the invariant formula complexity of $\Word_{G,k}$ with respect to the group $G^{k-1}$ acting on $\BAR G^k$ by
\[
  (g_1,\dots,g_{k-1})(M_1,\dots,M_k)
  =
  (M_1\BAR{g_1}^{-1},\BAR{g_1}M_2\BAR{g_2}^{-1},\dots,\BAR{g_{k-2}}M_{k-1}\BAR{g_{k-1}}^{-1},\BAR{g_{k-1}} M_k)
\]
where $\BAR{g_i}$ denotes the permutation matrix corresponding to $g_i$.
Note that the action $G^{k-1}$ on $\BAR{G}^k$ arises from an action on input variables $M_{i,a,b}$ 
given by
\[
  (g_1,\dots,g_{k-1})M_{i,a,b} = M_{i,g_{i-1}(a),g_i(b)}
\]
where $g_0=g_k=1_G$. Since this action is faithful, we may view $G^{k-1}$ as a subgroup of $S_{kn^2}$.

\subsection{The formula size of $\Word_{S_n,k}$}

The functions $\Word_{S_n,k}$ (also known as  {\em iterated permutation matrix multiplication}) have an important place in complexity theory, since languages $\{\Word_{S_5,k}\}_{k \in \mathbb N}$ and $\{\Word_{S_n,n}\}_{n \in \mathbb N}$ are respectively complete for complexity classes $\mathsf{NC^1}$ and $\mathsf{Logspace}$ \cite{barrington1989bounded,cook1987problems}. Understanding the formula complexity of the word problem for $S_n$ could eventually be key to separating these classes.

The smallest known DeMorgan formulas ($\AND_2,\OR_2,\NOT$ gates) for $\Word_{S_n,k}$ have size $n^{O(\log k)}$.  Allowing unbounded fan-in $\AND_\infty,\OR_\infty,\NOT$ gates (i.e.,\ $\mathsf{AC}$ formulas), we obtain the same size $n^{O(\log k)}$ with depth $O(\log k)$ and moreover by formulas that are syntactically $S_n^{k-1}$-invariant.  In fact, the smallest known $\mathsf{AC}$ formulas of any given depth (via a natural divide-and-conquer construction that we describe in \Cref{cor:upper}) happen to be $S_n^{k-1}$-invariant:
\[
  \mathcal L^{S_n^{k-1}}_{\mathsf{AC}}(\Word_{S_n,k}) 
  = 
  n^{O(\log k)},
  \qquad
  \mathcal L^{S_n^{k-1}}_{\mathsf{AC}_{d+1}}(\Word_{S_n,k}) 
  = 
  n^{O(d(k^{1/d}-1))}.
\]
(Note that the size-depth tradeoff on the right implies the formula size lower bound on the left, since $\lim_{d \to \infty} d(k^{1/d}-1) = \ln k$.)

A matching lower bound 
  $\mathcal L_{\mathsf{AC}}(\Word_{S_n,k}) 
  = 
  n^{\Omega(\log k)}$
for unrestricted (non-invariant) formulas would separate complexity classes $\mathsf{NC}^1$ and $\mathsf{Logspace}$.
This motivates the question of first showing matching lower bounds for $S_n^{k-1}$-invariant formulas.

\subsection{Our results}

Our first theorem shows that the above $S_n^{k-1}$-invariant upper bounds are nearly optimal, even for $\mathsf{TC}$ formula which include $\MAJ$ gates.

\begin{theorem}\label{Sn lower bound}
For all $n,k,d \ge 1$,
\[
  \mathcal L^{S_n^{k-1}}_{\mathsf{TC}}(\Word_{S_n,k}) \ge n^{\log_2 k},\qquad
  \mathcal L^{S_n^{k-1}}_{\mathsf{TC}_d}(\Word_{S_n,k}) \ge n^{d(k^{1/d}-1)}.
\]
\end{theorem}

(Note that our the righthand lower bound is off by $1$ in depth compared to the upper bound.)

Note that $\Word_{G,k}$ is simply the restriction of the function $\Word_{S_n,k}$ to the subdomain $\BAR{G}^k \subseteq \BAR{S_n}^k$.  
The $G^{k-1}$-invariant complexity of $\Word_{G,k}$ is moreover a nondecreasing function of $G$ in the subgroup lattice of $S_n$. 

Thus, Theorem \ref{Sn lower bound} follows from the main result of this paper (Theorem \ref{thm: simple lower bound} below), which lower bounds the $G^{k-1}$-invariant $\mathsf{TC}$ formula size of $\Word_{G,k}$ for certain subgroups $G \le S_n$. This is because the alternating subgroup $A_n\leq S_n$ is simple for $n\geq 5$ and has minimum faithful permutation representation of degree $n$.

This also motivates the question of strengthening Theorem \ref{Sn lower bound} by finding more proper subgroups $G \subseteq S_n$ with the same lower bound.

\begin{restatable}
{theorem}{simpleLB}\label{thm: simple lower bound}
Let $G$ be a finite simple group and suppose that $G \le S_n$ is a faithful permutation representation of minimum degree. Then
\[
  \mathcal L^{G^{k-1}}_{\mathsf{TC}}(\Word_{G,k}) \ge n^{\log_2 k},\qquad
  \mathcal L^{G^{k-1}}_{\mathsf{TC}_d}(\Word_{G,k}) \ge n^{d(k^{1/d}-1)}.
\]
\end{restatable}

Our proof of Theorem \ref{thm: simple lower bound} uses different arguments in the cases where $G$ is nonabelian (\Cref{sec:nonabelian simple}) and where $G$ is cyclic of order $p$ (\Cref{sec:prime order}). However, both cases use the same framework developed in \Cref{sec:framework}.

We remark that the {\em minimum degree} requirement is necessary in the nonabelian case, as the statement would be false for the regular representation $G \le S_{|G|}$. Minimum degree ensures every proper subgroup $H < G$ has index at least $n$; this fact plays a role in our lower bound.

An additional result of this paper gives a lower bound on the $C_q^{k-1}$-invariant $\mathsf{AC}$ formula size of the word problem for cyclic groups $C_q$ where $q$ is a prime power.
When $q$ is not prime, we do not know whether the technique of \Cref{thm: simple lower bound} yields a stronger lower bound for $\mathsf{TC}$ formulas.

\begin{restatable}
{theorem}{CyclicACLB}\label{thm: cyclic prime power LB}
Suppose that $q$ is a prime power and $C_q \le S_q$ is cyclic of order $q$. Then
\[
  \mathcal L^{C_q^{k-1}}_{\mathsf{AC}}(\Word_{C_q,k}) \ge q^{\log_2 k},\qquad
  \mathcal L^{C_q^{k-1}}_{\mathsf{AC}_d}(\Word_{C_q,k}) \ge q^{d(k^{1/d}-1)}.
\]
\end{restatable}

\begin{corollary}\label{cor:all finite groups}
Let $G$ be any finite group. Define
\begin{align*}
    n(G)&=\max_{H\leq G:H\text{ is simple.}}\min_{\pi:H\to S_m:\ker(\pi)=\{1_G\}}m,\qquad
    q(G)=\max_{C_q\leq G:\text{ $q$ is a prime power.}}q.
\end{align*}
Then 
\[
  \mathcal L^{G^{k-1}}_{\mathsf{TC_d}}(\Word_{G,k}) \ge n^{d(k^{1/d}-1)},\qquad
  \mathcal L^{G^{k-1}}_{\mathsf{AC}_d}(\Word_{G,k}) \ge q^{d(k^{1/d}-1)}.
\]
\end{corollary}

Note that \Cref{cor:all finite groups} is tight for abelian groups $G$ up to a factor of $|G|$, since for abelian $G$, we can write $G=C_{t_1}\times\dots\times C_{t_\ell}$ for some integers $\{t_i\}_{i\in [\ell]}$. This gives a divide and conquer $G^{k-1}$-invariant formula for $\Word_{G,k}$ by solving $\Word_{C_{t_i},k}$ for each $i$.

We remark that all of our lower bounds of the form $n^{d(k^{1/d}-1)}$ for $P$-invariant depth-$d$ formulas imply lower bounds $n^{k^{1/d}-1}/kn^2$ for $P$-invariant depth-$d$ circuits, since every ($P$-invariant) depth-$d$ circuit of size $s$ unfolds to a ($P$-invariant) depth-$d$ formula of size at most $s^{d-1} m$ where $m$ is the number of variables.

The parameters defined in \Cref{cor:all finite groups} have been studied before in the setting of all finite groups. Babai et al. \cite{babai1993faithful} show that for any finite group $G$, the minimum degree of a faithful permutation representation of $G$ upper bounds the size of every cyclic subgroup of prime power order. This shows that when $G$ is a finite simple group, the formula size lower bound given by \Cref{thm: simple lower bound} is always stronger than the bound given by \Cref{thm: cyclic prime power LB}. 

However, for some groups $G$ the lower bounds on $\mathcal{L}_{\mathsf{AC}}^{G^{k-1}}(\Word_{G,k})$ implied by \Cref{thm: simple lower bound} and \Cref{thm: cyclic prime power LB} are similar. In \Cref{sec:large gaps} we give examples where the gaps between the lower bounds may be large, conditional on some number-theoretic conjectures. Such examples justify the extra work to prove \Cref{thm: simple lower bound}, especially in the case where $G$ is nonabelian simple.

\subsection{The lower bound technique}

Let $f : \{0,1\}^m \to \{0,1\}$ be a $P$-invariant partial function, where $P \le S_m$ and $\Omega = \mathrm{Dom}(f)$ is a $P$-invariant subset of $\{0,1\}^m$.
To prove Theorem \ref{thm: simple lower bound}, we introduce a general framework that lower bounds the $P$-invariant $\mathsf{TC}_d$ formula size of $f$ in terms of a function $\beta_{Q,\Omega,d}(P)$ where $Q$ is any supergroup $P \le Q \le S_m$ such that $P$ is the $Q$-stabilizer of $f$.
Functions $\beta_{Q,\Omega,d}$ may be viewed as complexity measures on pairs of subgroups of $Q$, which are interesting in their own right. These complexity measures correspond to ``formulas'' that construct $P$ in a certain manner starting from $Q$-stabilizers of points $1,\dots,m$.

This lower bound framework is described in \Cref{sec:framework}. It is then applied to $\Word_{G,k}$ in for nonabelian simple groups in \Cref{sec:nonabelian simple}, and for cyclic group of prime order in \Cref{sec:prime order}.

For a broad outline of our technique see \Cref{sec:technique}

Our proof of Theorem \ref{thm: cyclic prime power LB} in \Cref{sec:prime power} uses a different argument that generalizes a proof in the special case $q=2$ in previous work of the second author \cite{Rossman2018}.

\subsection{Related work}

\subsubsection{$\mathsf{TC}_d$ and $\mathsf{AC}_d$ formula lower bounds for $\Word_{S_n,k}$}

For context, we state below the strongest known lower bounds for $\Word_{S_n,k}$ with respect to bounded depth $\mathsf{AC}$ formulas (without the restriction of $S_n^{k-1}$-invariance). The first three results all use Switching Lemmas, which limits these lower bounds to depth $o(\log n + \log k)$. In contrast, our asymptotically tight lower bounds in the $S_n^{k-1}$-invariant setting extend to arbitrary depth (but stabilize at depth $O(\log k)$).
\begin{itemize}
\item
Rossman \cite{Rossman2018a} gives a lower bound $n^{\Omega(\log k)}$ on the $\mathsf{AC}$ formula size of $\Word_{S_n,k}$ when $k \le \log\log n$, however only up to depth $d \le {\log n}/{(\log\log n)^{O(1)}}$.
\item
Beame, Impagliazzo and Pitassi \cite{Beame1998} prove a size-depth tradeoff $n^{\Omega(k^{1/\exp(d)})}$ for depth-$d$ $\mathsf{AC}$ formulas computing $\Word_{S_n,k}$ when $k \le \log n$ and $d \le \log\log k$.
\item
Chen, Oliveira, Servedio and Tan \cite{Chen2016} give an improved tradeoff $n^{\Omega(k^{1/2d}/d)}$
when $k \le n^{1/5}$ and $d \le \log k/\log\log k$,
however not for the function $\Word_{S_n,k}$ (a.k.a.\ iterated permutation matrix multiplication) but rather for the more general {\em iterated Boolean matrix multiplication} problem where input matrices $M_1,\dots,M_k$ are not guaranteed to be permutation matrices.
\item
In the case $n=2$, where $\Word_{S_2,k}$ reduces to the $\textsc{Parity}_k$, lower bound of \cite{Hastad1986,khrapchenko1971complexity,Rossman2015a} imply asymptotically tight tradeoffs:
\[
  \mathcal L_{\mathsf{AC}}(\Word_{S_2,k}) = \Theta(k^2),\qquad
  \mathcal L_{\mathsf{AC}_{d+1}}(\Word_{S_2,k}) = 2^{\Theta(d(k^{1/d}-1))}.
\]
Since these bounds are merely polynomial, they fail to separate $\mathsf{NC}^1$ from $\mathsf{Logspace}$.
\item
Rossman \cite{Rossman2018} gives nearly tight bounds on the $S_2^{k-1}$-invariant formula size of $\Word_{S_2,k}$:
\[
  2^{d(k^{1/d}-1)}
  \le
  \mathcal L^{S_2^{k-1}}_{\mathsf{AC}_{d+1}}(\Word_{S_2,k}) 
  \le
  k2^{dk^{1/d}}.
\]

\item 
Impagliazzo, Paturi, and Saks \cite{impagliazzo1997size} prove that a depth-$d$ $\mathsf{TC}$ circuit computing $\Word_{S_2,k}$ must have size at least $k^{1+1/(1+\sqrt2)^d}$. By hardness-magnification results of Chen and Tell \cite{chen2019bootstrapping} expanding on work by Allender and Kouck\'{y} \cite{allender2010amplifying}, there exists $c>1$ such that improving the above bound to $k^{1+c^{-d}}$ would imply the separation $\mathsf{TC}^0\neq \mathsf{NC}^1$.

\item
Using representation theory, Alexeev, Forbes, Tsimerman \cite{alexeev2011tensor} prove upper bounds on the rank of the group tensor (the polynomial version of the word problem), hence proving upper bounds on arithmetic formula complexity of the word problem. They show that the depth-$d+1$ arithmetic formula size of the word problem polynomial $W_{G,k}$ is at most $\sum_{\rho\in \widehat{G}}\deg(\rho)^{dk^{\frac1d}}$, where $\widehat{G}$ is the set of irreducible representations of $G$. However, for many groups, their formulas are not $G^{k-1}$-invariant due to making choices of cosets of subgroups.
\end{itemize}

In this paper we improve on the result of \cite{Rossman2018} in two ways. First, in the setting of cyclic groups of prime order, we extend the lower bound from formulas in $\mathsf{AC}$ to formulas in $\mathsf{TC}$ as stated in \Cref{thm: simple lower bound}. Second, we generalize the arguments used to cyclic groups of prime power order for the $\mathsf{AC}$ lower bound in \Cref{thm: cyclic prime power LB}.

\subsubsection{The word problem for finitely presented groups}\label{sec:fp groups}

An important point to make is the distinction between our version of the word problem on groups and other more commonly studied versions. In the most commonly studied version of the word problem, one is given a set of generators for the group and the relations among the generators. Then the problem is to determine whether a sequence of these generators multiplies to equal the identity element. Often these problems deal with infinite groups. The complexity of this version of the word problem has been well-studied \cite{Dehn1911,Bartholdi2020,Myasnikov2017,Miasnikov2018}. On the other hand, we only deal with finite groups and do not worry about generators and relations, but are rather explicitly given the group elements serving as input to our problem.

\section{Preliminaries}

We fix the following notation throughout. For any $k\in \mathbbm{N}$, let $[k]=\{1,\dots,k\}$.

$G$ shall always be a finite permutation group on $n$ elements (i.e, a subgroup of $S_n$), which we assume to be transitive. $P$ and $Q$ shall always be subgroups of $S_m$, where $m = kn^2$ in our application. 

We write $1_G$ for the identity element in $G$ and 
$Z(G)$ for the center of $G$. 
If $H$ is a subgroup of $G$, then we write $H\leq G$. If moreover $H$ is a normal subgroup, then we write $H\trianglelefteq G$. If $G$ has no nontrivial proper normal subgroup, then $G$ is \emph{simple}. We denote $|G|$ to be the order of $G$.

We will often consider $k$-fold powers of a group $G$, denoted $G^k$. For $S\subseteq[k]$, let $\pi_S:G^k\to G^S$ be the projection homomorphism to $S$. A {\em permutation representation} of $G$ is a homomorphism $G\to S_n$ for some positive integer $n$. We call $n$ the {\em degree} of the representation, and if the homomorphism is injective, then the representation is called {\em faithful}. Define the subgroup $\Diag(G^2)<G^2$ to be the subgroup $\{(g,g):g\in G\}$.

\begin{definition}\normalfont
A {\em labeled tree} is a finite unordered rooted tree in which each node is associated with a label.  

An {\em $m$-variable $\AC$ formula} (respectively, {\em $\TC$ formula}) is a labeled tree in which each leaf (``input'') is labeled by a constant symbol or literal in the set $\{\texttt 0,\texttt 1,\texttt x_1,\BAR{\texttt x_1},\dots,\texttt x_m,\BAR{\texttt x_m}\}$ and each non-input (``gate'') is labeled by $\AND$ or $\OR$ (respectively, labeled by $\MAJ$).

The {\em depth} of a formula is the maximum number of gates on a leaf-to-root path. The {\em size} (a.k.a.\ {\em leaf-size}) of formula is the number of leaves that are labeled by literals.

For $d \in \mathbbm N$, we write $\AC_d$ (resp.\ $\TC_d$) for the set of depth-$d$ $\AC$ formulas (resp.\ $\TC$ formulas) {\em up to isomorphism}.  That is, we consider two formulas to be the same iff there exists an isomorphism between them as labeled graphs.

An alternative, concrete inductive definition of sets $\AC_d$ and $\TC_d$ is given by:
\begin{itemize}
    \item 
      Let $\AC_0 = \TC_0 = \{\texttt 0,\texttt 1,\texttt x_1,\BAR{\texttt x}_1,\dots,\texttt x_m,\BAR{\texttt x}_m\}$.
    \item
      For $d \ge 1$, 
      $\AC_d$ (resp.\ $\TC_d$) is the set of pairs of form $(\AND,I)$ or $(\OR,I)$ (resp.\ of the form $(\MAJ,I)$) where $I$ is a multiset of formulas in $\AC_0 \cup \dots \cup \AC_{d-1}$ (resp.\ $\TC_0 \cup \dots \cup \TC_{d-1}$).
\end{itemize}

The symmetric group $S_m$ acts on $\AC_d$ and $\TC_d$ by permuting indices on literals. For $P \le S_m$, we say that a formula $\Phi$ is {\em $P$-invariant} if $\pi \Phi = \Phi$ for all $\pi \in P$.

Every formula $\Phi$ computes a Boolean function denoted $\sem{\Phi} : \{0,1\}^m \to \{0,1\}$ in the usual way.  
For a partial function $f : \Omega \to \{0,1\}$ where $\Omega \subseteq \{0,1\}^m$, we say that $\Phi$ {\em computes} $f$ if $\sem{\Phi}(x)=f(x)$ for all $x \in \Omega$.

Note that $P$-invariance of $\Phi$ implies $P$-invariance of $\sem{\Phi}$, but the converse need not hold in general.  Also note that $\MAJ$ gates can simulate both $\AND$ and $\OR$ gates (by padding by an appropriate number of zeros or ones).  Any function computable by $\AC$ formulas of a given size and depth is therefore computable by a $\TC$ formula of the same size and depth.
\end{definition}

\section{$S_n^{k-1}$-invariant formulas for $\Word_{S_n,k}$}\label{Upper bounds}

\begin{lemma}\label{lem:upper}
For all $n,k,d \ge 1$ such that $k^{1/d}$ is an integer, the function $\Word_{S_n,k}$ is computed by $S_n^{k-1}$-invariant
$\Sigma_{d+1}$ and $\Pi_{d+1}$ formulas of size $kn^{d(k^{1/d}-1)}$.
\end{lemma}

\begin{proof}
For $u_0,u_k \in [n]$, let $\Word^{(u_0,u_k)}_{S_n,k}(M_1,\dots,M_k)$ denote the $(u_0,u_k)$-entry of $M_1\cdots M_k$. Thus, $\Word_{S_n,k}$ is the function $\Word^{(1,1)}_{S_n,k}$.

In the base case $d=1$, we have $\Sigma_2$ and $\Pi_2$ formulas:
\begin{align*}
  \Word&{}^{(u_0,u_k)}_{S_n,k}(M_1,\dots,M_k)\\
  &=
  \bigvee_{u_1,u_2,\dots,u_{k-1} \in [n]}\ 
  \Big(
    M_{1,u_0,u_1} \wedge M_{2,u_1,u_2} \wedge \dots \wedge M_{k-1,u_{k-2},u_{k-1}} \wedge M_{k,u_{k-1},u_k}
  \Big),\\
  \Word&{}^{(u_0,u_k)}_{S_n,k}(M_1,\dots,M_k) \\
  &=
  \bigwedge_{u_1,u_2,\dots,u_{k-1} \in [n]}\ 
  \Big(
    \Big(M_{1,u_0,u_1} \wedge M_{2,u_1,u_2} \wedge \dots \wedge M_{k-1,u_{k-2},u_{k-1}}\Big) \Rightarrow M_{k,u_{k-1},u_k}
  \Big)\\
  &=
  \bigwedge_{u_1,u_2,\dots,u_{k-1} \in [n]}\ 
  \Big(
    \bigvee_{j\in[k-1]}
    \neg M_{j,u_{j-1},u_j}
    \vee M_{k,u_{k-1},u_k}
  \Big).
\end{align*}
Both formulas are clearly $S_n^{k-1}$-invariant and have size $kn^{k-1}$.

For the induction step, assume $d \ge 2$ and let $\ell = k^{(d-1)/d}$. Then
\begin{align*}
  \Word&{}_{S_n,k}^{(u_0,u_k)}(M_1,\dots,M_k)\\ 
  &=
  \bigvee_{u_\ell,u_{2\ell},\dots,u_{k-\ell} \in [n]}\ 
  \bigwedge_{j\in[k^{1/d}]}
  \Word_{S_n,\ell}^{(u_{(j-1)\ell},u_{j \ell})}(M_{(j-1)\ell+1},\dots,M_{j\ell}).
\end{align*}
By induction, each $\Word_{S_n,\ell}^{(\cdot,\cdot)}$ has an $S_n^{\ell-1}$-invariant $\Pi_{d-1}$ formula of size $\ell n^{(d-1)(\ell^{1/(d-1)}-1)}$. Substituting these formulas above and collapsing the adjacent layers of $\AND$ below the output, we get an $S_n^{k-1}$-invariant $\Sigma_d$ formula for $\Word_{S_n,k}^{(u_0,u_k)}(M_1,\dots,M_k)$ of size $n^{(k-\ell)/\ell} k^{1/d} \cdot \ell n^{(d-1)(\ell^{1/(d-1)}-1)}$, which equals $kn^{d(k^{1/d}-1)}$.
We similarly get an $S_n^{k-1}$-invariant $\Pi_d$ formula from the observation that
\begin{align*}
  \Word&{}_{S_n,k}^{(u_0,u_k)}(M_1,\dots,M_k)\\ 
  &=
  \bigwedge_{u_\ell,u_{2\ell},\dots,u_{k-\ell} \in [n]}
  \Big(
  \bigvee_{j\in[k^{1/d}-1]}
  \neg\Word_{S_n,\ell}^{(u_{(j-1)\ell},u_{j \ell})}(M_{(j-1)\ell+1},\dots,M_{j\ell})\\
  &\hspace{2.5in} \vee
  \Word_{S_n,\ell}^{(u_{k-\ell},u_{k})}(M_{k-\ell+1},\dots,M_{k})
  \Big).
\end{align*}
Here we replace each $\Word_{S_n,\ell}^{(\cdot,\cdot)}$ subformula under a negation with a $\Pi_{d-1}$ formula, which we then convert to $\Sigma_{d-1}$ formula using DeMorgan's law.
\end{proof}

When $k^{1/d}$ is not necessary an integer, a similar divide-and-conquer construction gives the following upper bound:

\begin{corollary}\label{cor:upper}
For all $n,k,d \ge 1$, the function $\Word_{S_n,k}$ is computed by $S_n^{k-1}$-invariant $\Sigma_{d+1}$ and $\Pi_{d+1}$ formulas of size $n^{O(dk^{1/d})}$.  In particular, 
$\Word_{S_n,k}$ has $S_n^{k-1}$-invariant $\mathsf{AC}$ formulas of size $n^{O(\log k)}$ and depth $O(\log k)$.
\end{corollary}

The upper bound of Corollary \ref{cor:upper} is the quantitatively strongest known (up to constants in the exponent), even for $\mathsf{TC}$ formulas that may include $\MAJ$ gates.

\section{Overall technique}\label{sec:technique}
The structure of the proof of \Cref{thm: simple lower bound} is quite similar to that of the upper bound. Both bounds are inductive on formula depth, and deal with dividing the product $M_1 \cdots M_k$ into contiguous segments. While the upper bound keeps track of how many subproblems of permutation multiplication we have, the lower bound keeps track of the amount of symmetry that the subformulas of the formula must retain. This amount of symmetry is quantified by the length of the longest contiguous segment of the product (i.e. $M_i \cdots M_j$ for some $1 \leq i \leq j \leq k$) over which the formula is fully symmetric. 

Before giving the proof of \Cref{thm: simple lower bound}, we give a rough overview of our proof using the case in which $G=A_n$, the alternating group on $n$ elements. However, we remark that there is no significant simplification in the proofs given in \Cref{sec:nonabelian simple} by setting $G=A_n$ rather than allowing $G$ to be any nonabelian simple group. We merely focus on $A_n$ here for concreteness of the permutation action on $[n]$.

Let $\Phi$ be an $A_n^{k-1}$-invariant formula computing $\Word_{A_n,k}$. As a first step, we let $A_n^{k-1}$ act on $\Phi$ as a subgroup of a larger group action, given by letting $Q=A_n^{2k}$ act on inputs $(h_1,\dots,h_k)$:
\begin{align*}
    (g_1,\dots,g_{2k})(h_1,\dots,h_k)&=(g_1h_1g_2^{-1},\dots,g_{2k-1}h_kg_{2k}^{-1}).
\end{align*}
Now $A_n^{k-1}$ can be viewed as the subgroup of $Q=A_n^{2k}$ given by the equations $g_2=g_3,\dots,g_{2k-2}=g_{2k-1},g_{2k}=g_1=1_G$. 

We will define a complexity measure $\mu$ on subgroups of $Q$. Roughly, if a formula is stabilized by a subgroup $H$, then $\mu(H)$ provides a good lower bound on the size of that formula. In this overview, we will not say what exactly this complexity measure is, but we will say that it quantifies segments of the product $M_1\cdots M_k$ over which a formula is fully symmetric. What is more important here is the following two properties of $\mu$.
\begin{enumerate}
    \item Let a set of formulas $\{\Phi_j\}$ have $Q$-stabilizer $H$. If $\{\Phi_j\}=O_1\cup\dots\cup O_t$ breaks this set into its orbits under $H$, then there is some orbit $O_i$ with $Q$-stabilizer $H'$ such that $\mu(H')\geq \mu(H)$.
    \item Let $\Phi_j$, which has $Q$-stabilizer $H_j$, lie in the orbit $O_i$ stabilized by $H'$. Then $\mu(H_j)\geq \frac{\mu(H')}{1+\log_n[H':H'\cap H_j]}$.
\end{enumerate}
The idea for using these properties to get formula lower bounds is as follows. We can let $Q$ act on a formula $\Phi=(\MAJ,\{\Phi_j\})$ and find its stabilizer $H=\Stab_Q(\Phi)$. We then break $\{\Phi_j\}$ into orbits under the action of $H$: $\{\Phi_j\}=O_1\cup\dots\cup O_t$. From the first property we know that one of the orbits $O_i$ has a high complexity stabilizer $H'$.  

We examine a representative $\Phi_j\in O_i$ of that high complexity orbit. At this point we have roughly two things that can happen. First, if the $Q$-stabilizer $H_j=\Stab_Q(\Phi_j)$ of $\Phi_j$ has $[H':H'\cap H_j]$ small then it must be that $\mu(H_j)$ is large by the second property, and we apply induction to get a good lower bound on the size of $\Phi_j$ in terms of $\mu(H_j)$. Each subformula in the orbit $O_i$ has the same size as $\Phi_j$ by transitivity, so this gives a good lower bound on the size of $\Phi$. Otherwise, $[H':H'\cap H_j]$ is large, and the orbit-stabilizer theorem tells us that there are many subformulas in $O_i$, which again provides a good lower bound on the size of $\Phi$ in terms of $\mu(H)$.

Finally, we show that $\mu(A_n^{k-1})$ is large in order to get a good lower bound on the size of a formula with $Q$-stabilizer equal to $A_n^{k-1}$. An $A_n^{k-1}$-invariant formula computing the word problem must have $Q$-stabilizer equal to $A_n^{k-1}$, so we have the lower bound we want.

This is an oversimplified overview of our actual analysis, and the real proof of the lower bound considers pairs of subgroups corresponding to formulas. One subgroup in the pair is the stabilizer of the formula, and the other is the stabilizer of the function computed by the formula. This can be defined in a formal way. More care also needs to be put into the second property, especially in concerns such as the transition from $H'$-stabilizers to $Q$-stabilizers. Such considerations complicate the analysis, but the above ideas still lie at the core of the argument.

\section{The lower bound framework}\label{sec:framework}
Let $Q \le S_m$, and let $\Omega$ be a $Q$-invariant subset of $\{0,1\}^m$.
We consider the action of $Q$ on the set of functions with $\Omega$ (and arbitrary codomain).

For an integer-valued function $f:\Omega\to\mathbbm{N}$, let 
\[
  \mathrm{Stab}_Q(f) = \{\pi \in Q : f(x) = f(x_{\pi(1)},\dots,x_{\pi(m)}) \text{ for all } x \in \Omega\}.
\]
For a set of formulas $\{\Phi_i\}$ with each $\Phi_i$ taking inputs in $\Omega$, let
\begin{align*}
    \Stab_Q(\{\Phi_i\})=\{\pi\in Q: \{\pi \Phi_i\}=\{\Phi_i\}\}.
\end{align*}
To simplify notation, write $\Stab_Q(\Phi)=\Stab_Q(\{\Phi\})$ when $\{\Phi\}$ is a set containing just one formula.

Let $\mathcal B_{Q,\Omega} \subseteq \mathcal N_{Q,\Omega}$ be the following sets of subgroups of $Q$:
\begin{align*}
  \mathcal{B}_{Q,\Omega} &= \{\Stab_Q(f): f\text{ is a Boolean function with domain }\Omega \text{ and codomain } \{0,1\}\},\\
  \mathcal N_{Q,\Omega} &= \{\mathrm{Stab}_Q(f) : f \text{ is a function with domain } \Omega \text{ and codomain $\mathbbm{N}$}\}.
\end{align*}

Note $\mathcal{N}_{Q,\Omega}$ is the closure of $\mathcal{B}_{Q,\Omega}$ under intersections.

Let $\chi_1,\dots,\chi_m : \{0,1\}^m \to \{0,1\}$ be the coordinate functions $\chi_i(x) = x_i$. For a formula $\Phi$, we suppress notation and write $\sem{\Phi}|_\Omega=\sem{\Phi}$ when $\Omega$ is clear from context.

\begin{definition}\normalfont\label{def:beta}
For $d=0,1,2,\dots$, we define a function $\beta_{Q,\Omega,d} : \{(H,K):H\leq K\leq Q, K\in\mathcal{B}_{Q,\Omega}\} \to \mathbbm{N} \cup \{\infty\}$ inductively. For any $H\leq K\leq Q$ with $K\in\mathcal{B}_{Q,\Omega}$, let 
\[
  \beta_{Q,\Omega,0}(H,K) = \begin{cases}
    0 &\text{if } H = Q,\\
    1 &\text{if } H < Q \text{ and $H = \mathrm{Stab}_Q(\chi_i|_\Omega)$ for some $i \in [m]$},\\
    \infty &\text{otherwise.}
  \end{cases}
\]
Then, for $d \ge 1$ define
\begin{align*}
    \beta_{Q,\Omega,d}(H,K)
    &= 
    \min_{\substack{
    \text{$(H,K)$-good }r\in\mathbbm{N},\\
    H_1,\dots,H_r\leq Q,\\
    U_1,\dots,U_r\leq Q,\\
    L_1,\dots,L_r\in\mathcal{N}_{Q,\Omega},\\
    V_1,\dots,V_r\in\mathcal{B}_{Q,\Omega}.
    }}\max_{i\in[r]}[H_i:H_i\cap U_i] \beta_{Q,\Omega,d-1}(U_i,V_i),
\end{align*}
where $(r,(H_i)_{i\in[r]},(U_i)_{i\in[r]},(L_i)_{i\in[r]},(V_i)_{i\in[r]})$ is \emph{$(H,K)$-good} if and only if all of the following hold.
\begin{enumerate}[label=\alph*)]
    \item\label{prop a} For all $i\in[r]$, $U_i\leq V_i$ and $H_i\leq L_i$.
    \item\label{syntactic intersection} $H_1\cap\dots \cap H_r=H$.
    \item\label{semantic intersection} $L_1\cap\dots \cap L_r\leq K$.
    \item\label{syntactic conjugate} For all $i\in[r]$, $\bigcap_{h\in H_i}h^{-1}U_ih\leq H_i$.
    \item\label{semantic conjugate} For all $i\in[r]$, $\bigcap_{h\in H_i}h^{-1}V_ih\leq L_i$. 
\end{enumerate}
Let $\beta_{Q,\Omega}(H,K) = \lim_{d\to\infty} \beta_{Q,\Omega,d}(H,K)$.
\end{definition}

Note that in the inductive definition of $\beta_{Q,\Omega,d}$, the quantity $\beta_{Q,\Omega,d-1}(U_i,V_i)$ is defined because $V_i\in\mathcal{B}_{Q,\Omega}$ and \Cref{prop a}. Also observe that monotonicity with respect to $d$ is displayed: $\beta_{Q,\Omega,0}(H,K) \ge \beta_{Q,\Omega,1}(H,K) \ge \beta_{Q,\Omega,2}(H,K) \ge\dots$. Also, we have that $\beta_{Q,\Omega,d}(H,Q)=0$ for all $d$ and $H\leq Q$.

\Cref{lem:framework} shows that given some $Q$ acting on $\Omega$, $\beta_{Q,\Omega,d}(H,K)$ lower bounds the size of formulas $\Phi$ with $\Stab_Q(\Phi)=H$ and $\Stab_Q(\sem{\Phi})=K$. This is helpful because it converts the problem of lower bounding formula size to a purely group-theoretic problem of lower bounding the inductively defined function $\beta_{Q,\Omega,d}$ on pairs of subgroups. 

Before stating and proving the result, we remark that it is important that we defined $\beta_{Q,\Omega,d}$ on pairs of subgroups of $Q$, one of which is in $\mathcal{B}_{Q,\Omega}$. This is because the sets $\mathcal{B}_{Q,\Omega}$ and $\mathcal{N}_{Q,\Omega}$ may not contain $\Stab_Q(\Phi)$ for a formula $\Phi$. In general, formulas may carry more information about themselves than just the way they evaluate on inputs in $\Omega$.

However, it is often the case that $\mathcal{N}_{Q,\Omega}$ is easier to study than the set of all subgroups of $Q$, which may contain badly behaved subgroups. Thus, in our applications, our complexity measures also may take into account $\Stab_Q(\sem{\Phi})$, which must lie in $\mathcal{B}_{Q,\Omega}$, since $\sem{\Phi}$ is a Boolean function on $\Omega$. We then leverage the relationship $\Stab_Q(\Phi)\leq \Stab_Q(\sem{\Phi})$ frequently.

\begin{lemma}\label{lem:framework}
Let $\Phi$ be a $\mathsf{TC}_d$ formula. Then
\[
  \mathrm{size}(\Phi) \ge \beta_{Q,\Omega,d}(\mathrm{Stab}_Q(\Phi),\Stab_{Q}(\sem{\Phi})).
\]
\end{lemma}
\begin{proof}
Let $H=\Stab_Q(\Phi)$ and $K=\Stab_Q(\sem{\Phi})$. First note that $\beta_{Q,\Omega,d}(H,K)$ is well-defined because $H\leq K$ and $K$ is the $Q$-stabilizer of the Boolean function $\sem{\Phi}$ on $\Omega$. We prove by induction on $d$ that $\mathrm{size}(\Phi)\geq \beta_{Q,\Omega,d}(H,K)$. 
The base case $d=0$ follows from the definitions of $\beta_{Q,\Omega,0}$ and $\mathsf{TC}_0$.

Now assume that the result holds for all formulas up to depth $d$. Let $\Phi$ be a $\mathsf{TC}_{d+1}$ formula. Assume that $\Phi=(\MAJ,I)$, where $I$ is a multiset of formulas in $\mathsf{TC}_0\cup\dots\cup \mathsf{TC}_d$. The group $H=\Stab_Q(\Phi)$ then acts on $I$ as a permutation group, since $\Phi$ is $H$-invariant. Break $I$ into its $r$ orbits under the action of $H$.  

For $i\in[r]$, let $\{\Psi_{i,j}:j\in[t_i]\}$ be the $i$th orbit of $I$ under the action of $H$. Let $H_i=\Stab_Q(\{\Psi_{i,j}:j\in[t_i]\})$. Let $U_i=\Stab_Q(\Psi_{i,1})$ and let $V_i=\Stab_Q(\sem{\Psi_{i,1}})$. Finally, let $L_i=\Stab_Q(f_i)$, where $f_i$ is the integer-valued function  $f_i:\Omega\to \mathbbm{N}$ defines by $f_i(x)=|\{j:\Psi_{i,j}(x)=1\}|$ for $x\in\Omega$. Note that this is the Hamming weight function on the wires from this orbit of subformulas.

We claim that the tuple $(r,(H_i)_{i\in[r]},(U_i)_{i\in[r]},(L_i)_{i\in[r]},(V_i)_{i\in[r]})$ is $(H,K)$-good. We prove that this tuple satisfies each one of the conditions stated in the definition below.

For all $i\in[r]$, we have $U_i\leq V_i$ and $H_i\leq L_i$ because for any formula $\Psi$, we have $\Stab_Q(\Psi)\leq\Stab_Q(\sem{\Psi})$, so
\begin{align*}
    U_i= \Stab_Q(\Psi_{i,1})\leq\Stab_Q(\sem{\Psi_{i,1}})=V_i.
\end{align*}
Moreover, for each $i\in[r]$, we have 
\begin{align*}
    H_i=\Stab_Q(\{\Psi_{i,j}:j\in[t_i]\})\leq\Stab_Q(f_i)=L_i.
\end{align*}
That $V_i$ is an element of $\mathcal{B}_{Q,\Omega}$, and $L_i$ is an element of $\mathcal{N}_{Q,\Omega}$ follows because $V_i$ is the $Q$-stabilizer of the Boolean function $\sem{\Psi_{i,1}}$ on $\Omega$, while $L_i$ is the $Q$-stabilizer of the integer-valued function $f_i$ on $\Omega$. This proves that \Cref{prop a} is satisfied.

Now we prove that $H_1\cap \dots \cap H_r=H$. First, we have
\begin{align*}
    \bigcap_{i\in[r]}H_i&=\bigcap_{i\in[r]}\Stab_Q(\{\Psi_{i,j}:j\in[t_i]\})\leq \Stab_Q((\MAJ,\bigcup_{i\in[r]}\{\Psi_{i,j}:j\in[t_i]\}))=\Stab_Q(\Phi)=H.
\end{align*}
That $H\leq H_1\cap \dots\cap H_r$ follows because orbits of $I$ under the action of $H$ are $H$-invariant. This proves that \Cref{syntactic intersection} is satisfied.

If $q\in Q$ is such that $f_i(qx)=|\{j\in[t_i]:\Psi_{i,j}(qx)=1\}|=|\{j\in[t_i]:\Psi_{i,j}(x)=1\}|=f_i(x)$ for all $x\in \Omega$ and $i\in[r]$, then $\sem{\Phi}(qx)=\sem{\Phi}(x)$ for all $x\in\Omega$. Therefore, we have
\begin{align*}
    \bigcap_{i\in[r]}L_i=\bigcap_{i\in[r]}\Stab_Q(f_i)\leq \Stab_Q(\sem{\Phi})=K.
\end{align*}
This proves that \Cref{semantic intersection} is satisfied.

The $Q$-stabilizers of the subformulas $\Psi_{i,j}$ in $i$th orbit are all $H_i$-conjugates of each other. To see this, note that for every $\Psi_{i,j}$ there exists $h\in H_i$ such that $\Psi_{i,1}^h=\Psi_{i,j}$. Then, for any $s\in \Stab_Q(\Psi_{i,1})=U_{i}$, we have
\begin{align*}
    \Psi_{i,j}^{h^{-1}sh}&={\Psi_{i,1}^{sh}}={\Psi_{i,1}^h}={\Psi_{i,j}}.
\end{align*} 
This shows that $h^{-1}\Stab_Q({\Psi_{i,1}})h\leq \Stab_Q({\Psi_{i,j}})$. The symmetric argument shows the reverse inclusion. As a result, for each $i\in[r]$,
\begin{align*}
    \bigcap_{h\in H_i}h^{-1}U_ih&=\bigcap_{j\in[t_i]}\Stab_Q(\Psi_{i,j})\leq \Stab_Q(\{\Psi_{i,j}:j\in[t_i]\})=H_i.
\end{align*}
This proves that \Cref{syntactic conjugate} is satisfied.

\Cref{semantic conjugate} follows similarly, since for each $\Psi_{i,j}$ in the $i$th orbit, we have $\Stab_Q(\sem{\Psi_{i,j}})= h^{-1}\Stab_Q(\sem{\Psi_{i,}})h$ for some $h\in H_i$. To see this, again note that for every $\Psi_{i,j}$ there exists $h\in H_i$ such that $\Psi_{i,1}^h=\Psi_{i,j}$. Then, for any $s\in \Stab_Q(\sem{\Psi_{i,1}})=V_{i}$, we have
\begin{align*}
    \sem{\Psi_{i,j}}^{h^{-1}sh}&=\sem{\Psi_{i,j}^{h^{-1}sh}}=\sem{\Psi_{i,1}^{sh}}=\sem{\Psi_{i,1}^h}=\sem{\Psi_{i,j}}.
\end{align*}
This shows that $h^{-1}\Stab_Q(\sem{\Psi_{i,1}})h\leq \Stab_Q(\sem{\Psi_{i,j}})$. The symmetric argument shows the reverse inclusion. 

If $q\in Q$ is such that $\sem{\Psi_{i,j}}(qx)=\sem{\Psi_{i,j}}(x)$ for all $x\in \Omega$ and $i\in[r]$, then $f_i(qx)=f_i(x)$ for all $x\in\Omega$. As a result, for each $i\in[r]$,
\begin{align*}
    \bigcap_{h\in H_i}h^{-1}V_ih=\bigcap_{j\in[t_i]}\Stab_Q(\sem{\Psi_{i,j}})\leq \Stab_Q(f_i)= L_i.
\end{align*}
We have finished verifying that our tuple of $r$ and subgroups is $(H,K)$-good.

Recall that $\{\Psi_{i,j}:j\in[t_i]\}$ is the $i$th orbit of this group action stabilized by $H_i$, and $\Stab_Q(\Psi_{i,1})=U_i$. Then by the orbit-stabilizer theorem applied to the action of $H_i$ on the orbit, the size of the orbit is at least $[H_i:H_i\cap U_{i}]$. The size of each subformula in this orbit is equal to the size of $\Psi_{i,1}$, since for every $\Psi_{i,j}$, there exists $h\in H$ such that $\Psi_{i,j}=\Psi_{i,1}^h$, and $\mathrm{size}(\Psi_{i,1}^h)=\mathrm{size}(\Psi_{i,1})$ is clear. Therefore, for some $(r,(H_i)_{i\in[r]},(U_i)_{i\in[r]},(L_i)_{i\in[r]},(V_i)_{i\in[r]})$ that is $(H,K)$-good, we have
\begin{align*}
    \mathrm{size}(\Phi)&\geq \max_{i\in[r]}[H_i:H_i\cap U_i]\mathrm{size}(\Psi_{i,1})\\
    &\geq \max_{i\in[r]}[H_i:H_i\cap U_i]\beta_{Q,\Omega,d}(U_i,V_i)\\
    &\geq \beta_{Q,\Omega,d+1}(H,K).
\end{align*}
The second inequality follows by induction and monotonicity of $\beta_{Q,\Omega,d}$ in $d$. The last inequality follows from the inductive definition of $\beta_{Q,\Omega,d+1}$.
\end{proof}

To use the lemma to lower bound the size of $P$-invariant formulas, we start off with some $P$-invariant function $f : \Omega \to \{0,1\}$ where $\Omega$ is a $P$-invariant subset of $\{0,1\}^m$. To apply the lemma, find a supergroup $P \le Q \le S_m$ such that $\Omega$ is $Q$-invariant and $\Stab_Q(f) = P$. Then we get the lower bound $\mathrm{size}(\Phi)\geq\beta_{Q,\Omega,d}(P,P)$ for any $\Phi\in\mathsf{TC}_d$ with $\Stab_Q(\Phi)=\Stab_Q(\sem{\Phi})=P$.

Note that choosing $Q=P$ yields nothing, since $\beta_{P,\Omega,d}(H,P) = 0$ for all $d$ and $H\leq P$.

In our application, $G \le S_n$, $m = kn^2$, $P = G^{k-1}$, and $\Omega = \BAR{G}^k= \{(M_1,\dots,M_k) \in \{0,1\}^{kn^2}\} = \{0,1\}^m$, where the matrices are permutation matrices giving elements of $G\leq S_n$. 

Thus, we are left with the problem of lower bounding $\beta_{Q,\Omega,d}(P,P)$ for some choice of $Q$. This is a purely group-theoretic problem, and \Cref{lem:general lower bound} provides a framework to solve it.

\begin{lemma}\label{lem:general lower bound}
Assume $c\geq 1$ and that $\mu:\{(H,L):H\leq L\leq Q,L\in\mathcal{N}_{Q,\Omega}\}\to\mathbbm{N}$ is such that for any $H\leq K\leq Q$ with $K\in\mathcal{B}_{Q,\Omega}$,
\begin{enumerate}[label=\roman*),topsep=0.1in]
    \item\label{intersection property} Let $r\in\mathbbm{N}$, $H_1,\dots,H_r\leq Q$, $L_1,\dots,L_r\in\mathcal{N}_{Q,\Omega}$ be such that $H_i\leq L_i$ for all $i\in[r]$. Suppose that $H=H_1\cap\dots\cap H_r$ and $K\geq L_1\cap\dots\cap L_r$. Then there is some $i\in[r]$ such that $\mu(H_i,L_i)\geq \mu(H,K)$.
    \item\label{shrinkage property} Let $U\leq V\leq Q$ with $V\in\mathcal{B}_{Q,\Omega}$ and $L\geq H$ with $L\in\mathcal{N}_{Q,\Omega}$. Suppose that $\bigcap_{h \in H}h^{-1}Uh\leq H$ and $\bigcap_{h\in H}h^{-1}Vh\leq L$. Then $\mu(U,V)\geq \mu(H\cap U,V)\geq \frac{\mu(H,L)}{1+\log_c[H:H\cap U]}$.
    \item\label{literal property} Let $H=\Stab_{Q}(\chi_i|_\Omega)$ for some coordinate function $\chi_i$. Then $\mu(H,K)\leq 1$. 
    \item\label{Q property} $\mu(Q,Q)=0$. 
\end{enumerate}
Then for all $H\leq K\leq  Q$ with $K\in\mathcal{B}_{Q,\Omega}$, 
\begin{align*}
    \beta_{Q,\Omega,d}(H,K)&\geq  c^{d(\mu(H,K)^{1/d}-1)} \text{ for all } d \ge 1,\\
    \beta_{Q,\Omega}(H,K)&\geq  c^{\log_2(\mu(H,K))}.
\end{align*}
\end{lemma}
\begin{proof}
We prove by induction on $d$ and actually start at 0, interpreting for $d=0$
\begin{align*}
    c^{d(\mu(H,K)^{1/d}-1)}&=\begin{cases}
    0 &\text{if } \mu(H,K) = 0,\\
    1 &\text{if } \mu(H,K)=1,\\ 
    \infty &\text{otherwise.}
  \end{cases}
\end{align*}
Note that these interpretations fit into the inductive steps. The base case in which $d=0$ is clear by \Cref{literal property}, \Cref{Q property}, and the definition of $\beta_{Q,\Omega,0}$. Now assume that the result holds for $\beta_{Q,\Omega,i}$ with $i\leq d$.

By definition, for some tuple $(r,(H_i)_{i\in[r]},(U_i)_{i\in[r]},(L_i)_{i\in[r]},(V_i)_{i\in[r]})$ that is $(H,K)$-good we have
\begin{align*}
    \beta_{Q,\Omega,d+1}(H,K)&\geq \max_{i\in[r]}[H_i:H_i\cap U_i]\beta_{Q,\Omega,d}(U_i,V_i).
\end{align*}
Because the tuple is $(H,K)$-good, by \Cref{prop a}, \Cref{syntactic intersection}, and \Cref{semantic intersection}, we have that the hypotheses of \Cref{intersection property} are satisfied by the $H_i$ and $L_i$, so for some $i\in[r]$, we have $\mu(H_i,L_i)\geq \mu(H,K)$. 

Let $m=1+\log_{c}[H_i:H_i\cap U_i]$. By \Cref{prop a}, \Cref{syntactic intersection}, and \Cref{semantic conjugate} the hypotheses of \Cref{shrinkage property} are satisfied by $U_i$, $H_i$, $V_i$, and $L_i$, so we have $\mu(U_i,V_i)\geq \mu(H_i\cap U_i, V_i)\geq \frac{\mu(H_i,L_i)}{m}\geq \frac{\mu(H,K)}{m}$. Therefore,
\begin{align*}
    \beta_{Q,\Omega,d+1}(H,K)&\geq c^{m-1}\beta_{Q,\Omega,d}(U_i,V_i)
    \geq c^{m-1}c^{d((\frac{\mu(H,K)}{m})^{1/d}-1)}
    \geq c^{(d+1)(\mu(H,K)^{1/(d+1)}-1)}.
\end{align*}
The last step follows from optimization over $m$ using elementary calculus.

This lower bound approaches $c^{\ln(\mu(H,K))}$ as $d\to\infty$, but since $\mu$ takes integral values, when $d$ becomes large we actually have a lower bound $c^{\log_2(\mu(H,K))}$.
\end{proof}

We can informally interpret the results of this section in the following way. \Cref{lem:framework} shows that the size of a formula $\Phi$ with $\Stab_Q(\Phi)=H$ and $\Stab_Q(\sem{\Phi})=K$ depends on a sequence of operations on the subgroup lattice $\mathrm{Sub}(Q)$ of $Q$. We begin with a pair $(H,K)\in \mathrm{Sub}(Q)^2$. Then, to get to the pair stabilizing the orbits of subformulas of $\Phi$, we first ``go up'' in $\mathrm{Sub}(Q)^2$. 

That is, in the process of breaking the set of subformulas of $\Phi$ into $H$-orbits, we find that the syntactic and semantic stabilizers of these orbits are supergroups of $H$ and $K$ that satisfy the intersection properties given by \Cref{syntactic intersection} and \Cref{semantic intersection} stated. 

The next step is to find the stabilizer of a single subformula in an orbit. Here we apply the orbit-stabilizer theorem, and are ``going down'' in the subgroup lattice. At this point there is a cost associated to how far down on the subgroup lattice we go by the orbit-stabilizer theorem.

The point of \Cref{lem:general lower bound} is then to define a ``lower-bound witness'' $\mu$ on a subset of $\mathrm{Sub}(Q)^2$ that is accurately reflects the cost of these operations. A designer of a small invariant formula would then want to find sequences of subgroups (under the constraints stated in the definition of $\beta_{Q,\Omega,d}$) in the subgroup lattice to get to the low-complexity subgroups (stabilizers of coordinate functions and $Q$ itself) with as little cost as possible. Cost is accrued in the going down phase of the process. \Cref{lem:general lower bound} states conditions sufficient for the formula designer to not be able to easily decrease the lower-bound witness $\mu$.

\section{Applying the framework}
In this section we prove \Cref{thm: simple lower bound}. \Cref{sec:nonabelian simple} takes care of the nonabelian case and \Cref{sec:prime order} takes care of the abelian case.
\subsection{$\mathsf{TC}$ lower bounds for nonabelian simple groups}\label{sec:nonabelian simple}
Throughout this section let $G\leq S_n$ be a finite nonabelian simple permutation group, where $n$ is the the minimum degree of a faithful representation of $G$. We want to use the framework set up in \Cref{sec:framework} to prove a lower bound on the formula size of $G^{k-1}$-invariant $\mathsf{TC}_d$ formulas computing the function $\Word_{G,k} : \BAR{G}^k \to \{0,1\}$. 

To do so, we consider a larger group $G^{k-1} \le Q \le S_{kn^2}$ such that $\BAR{G}^k$ is $Q$-invariant and $\Stab_Q(\Word_{G,k}) = G^{k-1}$.
We then get the bound $\mathcal L^{G^{k-1}}_{\mathsf{TC}_d}(\Word_{G,k}) \ge \beta_{Q,\Omega,d}(G^{k-1},G^{k-1})$ by \Cref{lem:framework}, so for our purposes it suffices to lower bound $\beta_{Q,\Omega,d}(G^{k-1},G^{k-1})$
We do so by constructing a suitable $\mu$ satisfying the hypotheses of \Cref{lem:general lower bound} with $c=n$.

Recall that $n$ is the minimum degree of a faithful permutation representation of $G$. The following lemma characterizes this $n$ nicely.
\begin{lemma}\label{lem:min index}
$n=\min(\sqrt{|G|},\min_{A<G}[G:A])$.
\end{lemma}
\begin{proof}
That $\min(\sqrt{|G|},\min_{A<G}[G:A])=\min_{A<G}[G:A]$ follows from the classification of finite simple groups \cite{Wilson2009}. Then, $\min_{A<G}[G:A]=n$ follows from \cite{johnson1971minimal}.
\end{proof}

We first consider the ``left-right'' action of $G^{2k} \le S_{kn^2}$ on $\BAR{G}^k$ 
\begin{align*}
    (g_1,\dots,g_{2k})(M_1,\dots,M_k)&=(\BAR{g_1}M_1\BAR{g_2}^{-1},\dots,\BAR{g_{2k-1}}M_k\BAR{g_{2k}}^{-1}).
\end{align*}
Note that odd coordinates in $[2k]$ act on matrices on the left, while even coordinates act on the right. That is, the elements $g_{2i-1}$ and $g_{2i}$ act respectively on the left and right of the matrix $M_i$. 

The embedding of $G^{k-1}$ in $G^{2k}$ (as subgroups of $S_{kn^2}$) is given by the shifted diagonal action.
\begin{align}\label{embedding G}
    (g_1,\dots,g_{k-1})\mapsto (1_G,g_1^{-1},g_1,g_2^{-1},g_2,\dots,g_{k-1}^{-1},g_{k-1},1_G).
\end{align}
However, the $G^{2k}$-stabilizer of $\Word_{G,k}$ is larger than $G^{k-1}$, since it contains all elements of the form $(g_1,1_G,\dots,1_G)$ with $g_1\in\Stab_G(1)$ and $(1_G,\dots,1_G,g_{2k})$, where $g_{2k}\in \Stab_G(1)$.

For this reason, we take $Q$ to be the subgroup of $G^{2k}$ given by $Q=\{(g_1,\dots,g_{2k})\in G^{2k}:g_1=g_{2k}=1_G\}$. Now we have $G^{k-1} \le Q \le S_{kn^2}$ with $\Stab_Q(\Word_{G,k}) = G^{k-1}$, as required.

Observe that the coordinate functions $\chi_1,\dots,\chi_{kn^2}:\{0,1\}^{kn^2}\to\{0,1\}$ are given by the entries of the permutation matrices $\{M_1,\dots,M_k\}$.

We now define the $\mu$ that we will use in conjunction with \Cref{lem:general lower bound} to prove a lower bound on $\beta_{Q,\Omega,d}(G^{k-1},G^{k-1})$.

\begin{definition}\label{def: mu nonabelian simple}\normalfont
For $H\leq Q$ let 
\begin{align*}
    E(H) &=\{\{i,i+1\} : i\in [k-1] \text{ such that } H{\uhr}_{\{2i,2i+1\}}=\Diag(G^{\{2i,2i+1\}})\}.
\end{align*}
We can view $E(H)$ as the edge set of an undirected graph with vertex set $[k]$, which we denote $([k],E(H))$. Note that this graph is a spanning subgraph of the path graph with edge set $\{\{1,2\},\{2,3\},\dots,\{k-1,k\}\}$. 

For $H\leq K\leq Q$ with $K\in\mathcal{N}_{Q,\Omega}$, we define $\mu(H,K)$ to be the number of vertices in the largest connected component of $([k],E(H)\cap E(K))$. Define $\mu(H,Q)=0$ for any $H\leq Q$.
\end{definition}

We remark that this definition of $\mu$ is highly dependent on both $H$ and $L$. In fact, if $([k],E(H))$ has a large connected component, it still may be possible for $\mu(H,K)$ to be small for some choice of $L$. 

The consequence for our lower bound is as follows. Suppose $\Phi$ is a formula with $\Stab_Q(\Phi)=G^{k-1}$. Then $\Stab_Q(\Phi)=\Stab_Q(\Word_{G,k})$. However, we can say nothing about the size of $\Phi$ unless we also have some information about $\Stab_Q(\sem{\Phi})$. 

If $\Stab_Q(\sem{\Phi})=Q$, for example, then $\mu(\Stab_Q(\Phi),\Stab_Q(\sem{\Phi}))=0$ by definition. Then \Cref{lem:general lower bound} gives no meaningful lower bound on $\beta_{Q,\Omega,d}(\Stab_Q(\Phi),\Stab_Q(\sem{\Phi}))$, and hence we have no meaningful lower bound on $\mathrm{size}(\Phi)$. 

Such a $\Phi$ can be realized in the following way. Let $\Psi$ be such that $\sem{\Psi}=\Word_{G,k}$ and $\Stab_Q(\Psi)=G^{k-1}$. Let $\Phi=(\OR,\{\texttt{1},\Psi\})$. Then $\sem{\Phi}$ is constant, so $\Stab_Q(\sem{\Phi})=Q$. However, $\Stab_Q(\Phi)=\Stab_Q(\Psi)=G^{k-1}$. Thus, even though $\Stab_Q(\Phi)=G^{k-1}$ is a subgroup for which we expect nontrivial lower bounds, we have no lower bound on its size. Fortunately, such a formula $\Phi$ does not arise in an optimal construction, due to $Q$ being useless to a formula designer in the ``going up'' phase of formula design.

In the rest of this section, we show that once we take into account semantic stabilizers (as we have done in \Cref{def: mu nonabelian simple}), we do get nontrivial lower bounds. 

\subsubsection{Structural results}
Since $Q$ can be thought of as a $2k-2$-fold power of the group $G$, it is helpful to understand the structure of subgroups of direct powers of $G$. The following lemma of Goursat helps characterize these subgroups.
\begin{lemma}[\cite{Goursat1889}]
\label{lem: Goursat}
Let $A$ and $B$ be groups. Let $K$ be a subgroup of $A\times B$ such that the projections $\pi_A:K\to A$ and $\pi_B: K\to B$ are surjective. Let $M = \{a \in A : (a,1) \in K\} \trianglelefteq A$ and $N = \{b \in B : (1,b) \in K\} \trianglelefteq B$. Then there exists an isomorphism $\theta : A/M \to B/N$ such that $K = \{(a,b) \in A \times B : \theta(aM) = bN\}$.
\end{lemma}
This lemma is especially helpful in the special case where $A$ and $B$ are both nonabelian simple groups, since the possibilities for $M$ and $N$ are then restricted to either being the full group $A$ or $B$ respectively, or trivial.
\begin{corollary}\label{cor:diag to full}
Let $\Diag(G^2)<H\leq G^2$. Then $H=G^2$.
\end{corollary}
\begin{proof}
Since $H>\Diag(G^2)$ there exists $(g_1,g_2)\in H$ such that $g_1g_2^{-1}\neq 1_G$. Then, multiplying by $(g_2^{-1},g_2^{-1})\in \Diag(G^2)<H$, we find that $(g_1g_2^{-1},1_G)\in H$.

Note that by \Cref{lem: Goursat}, $\{1_G\}^{\{1\}}<H{\uhr}_{\{1\}}\trianglelefteq \pi_{\{1\}}(H)= G^{\{1\}}$. Because $G$ is nonabelian simple, we must have $H{\uhr}_{\{1\}}=G^{\{1\}}$. This also implies that $H{\uhr}_{\{2\}}=G^{\{2\}}$. 
\end{proof}

Another point at which \Cref{lem: Goursat} becomes useful is in our proof that $\mu$ satisfies \Cref{shrinkage property}. Note that for $\ell\geq 2$, if $H\leq Q$ is such that $E(H)$ induces a connected component of size $\ell$ in $([k],E(H))$, then $H$ has some shifted diagonal subgroup isomorphic to $G^\ell$. Thus, the following definition is very natural.

\begin{definition}\label{def:support}\normalfont
For a subgroup $H\leq G^k$, define a \emph{support} for $H$ to be a subset $T\subseteq[k]$ such that for all $i\not\in T$, $H{\uhr}_{\{i\}}=G^{\{i\}}$.
\end{definition}

\begin{lemma}\label{lem: quotient G}
For any $\ell\in\mathbbm{N}$, let $N\triangleleft G^\ell$ be such such that $\frac{G^\ell}{N}\cong G$. Then there is exactly one $j\in[\ell]$ such that $N{\uhr}_{\{j\}}<G^{\{\ell\}}$. 
\end{lemma}
\begin{proof}
For all $j\in[\ell]$ we have that $\pi_{\{j\}}(N)\trianglelefteq G^{\{j\}}$ by normality of $N$ in $G^\ell$. Therefore, this projection must be either the full group $G^{\{j\}}$ or trivial. 

There is at most one $j$ such that $\pi_{\{j\}}(N)=\{1_G\}^{\{j\}}$. Otherwise we would have $\pi_{\{j_1,j_2\}}(N)=\{1_G\}^{\{j_1,j_2\}}$, and 
\begin{align*}
    [G^\ell:N]\geq [\pi_{\{j_1,j_2\}}(G^\ell):\pi_{\{j_1,j_2\}}(N)]\geq |G|^2.
\end{align*}
The first inequality follows from the elementary group theory fact that for any group homomorphism $\varphi:A\to C$, we have $[\varphi(A):\varphi(B)]$ for any $B\leq A$.

Assume that there is no $j$ such that $\pi_{\{j\}}(N)=\{1_G\}^{\{j\}}$. Then there is at least one $j\in[\ell]$ such that $N{\uhr}_{\{j\}}=\{1_G\}^{\{j\}}$, since otherwise we would have $N{\uhr}_{\{j\}}=G^{\{j\}}$ for all $j\in[\ell]$. This follows from \Cref{lem: Goursat} implying that $N{\uhr}_{\{j\}}\trianglelefteq \pi_{\{j\}}(N)$ and simplicity of $G$. This would imply $N=G^\ell$, a contradiction. Without loss of generality assume that $j=\ell$.

Using \Cref{lem: Goursat}, write $N=\mathrm{Graph}(\pi_{\{\ell\}}(N)\cong\frac{\pi_{[\ell-1]}(N)}{N{\uhr}_{[\ell-1]}})$. Therefore, there exists some isomorphism $\varphi:\frac{\pi_{[\ell-1]}(N)}{N{\uhr}_{[\ell-1]}}\to G^{\{\ell\}}$ such that $N=\{(g_1,\dots,g_{\ell-1},g_\ell):\varphi(g_1,\dots,g_{\ell-1})N{\uhr}_{[\ell-1]}= g_\ell\}$. But such a subgroup cannot be normal in $G^\ell$, since for any pair of non-identity group elements $g\in G$, there exists $a\in G$ such that $a^{-1}ga\neq g$ by $G$ being nonabelian.

Thus, there is exactly one $j$ such that $\pi_{\{j\}}(N)=\{1_G\}^{\{j\}}$. Assume without loss of generality that $j=\ell$. Then $N=\mathrm{Graph}(\{1_G\}\cong\frac{\pi_{[\ell-1]}(N)}{N{\uhr}_{[\ell-1]}})$. Since $\pi_{[\ell-1]}(N)=G^{[\ell-1]}$, it must also be the case that $N{\uhr}_{[\ell-1]}=G^{[\ell-1]}$. Thus, this $j$ is the only such that $N{\uhr}_{j}< G^{\{\ell\}}$. 
\end{proof}

\begin{corollary}\label{cor: quotient G}
Let $N\triangleleft H\leq G^k$ be such that $\frac{H}{N}\cong G$. Let $T$ be a support for $H$. Then $N$ has a support of size exactly $|T|+1$. 
\end{corollary}
\begin{proof}
Without loss of generality assume that $T=\{k-|T|+1,\dots,k\}$. We claim that $\frac{H{\uhr}_{[k-|T|]}}{N{\uhr}_{[k-|T|]}}$ is isomorphic to a normal subgroup of $\frac{H}{N}$. To see this, first note that $N{\uhr}_{[k-|T|]}\trianglelefteq H{\uhr}_{[k-|T|]}$, since the intersection of a subgroup with a normal subgroup is normal.

Now let $(g_1,\dots,g_{k-|T|})N{\uhr}_{[k-|T|]}$ and $(h_1,\dots,h_{k-|T|})N{\uhr}_{[k-|T|]}$ be two distinct elements of $\frac{H{\uhr}_{[k-|T|]}}{N{\uhr}_{[k-|T|]}}$. We claim that $(g_1,\dots,g_{k-|T|},1_G^{|T|\mathrm{\:times}})N$ and $(h_1,\dots,h_{k-|T|},1_G^{|T|\mathrm{\:times}})N$ must be distinct elements of $\frac{H}{N}$.

Assume otherwise. Then $(g_1h_1,\dots,g_{k-|T|}h_{k-|T|},1_G^{|T|\mathrm{\:times}})$ is an element of $N$. This implies that $(g_1h_1,\dots,g_{k-|T|}h_{k-|T|})$ is an element of $N{\uhr}_{[k-|T|]}$. We have now contradicted that $(g_1,\dots,g_{k-|T|})N{\uhr}_{[k-|T|]}$ and $(h_1,\dots,h_{k-|T|})N{\uhr}_{[k-|T|]}$ are distinct.

This gives an injection of $\frac{H{\uhr}_{[k-|T|]}}{N{\uhr}_{[k-|T|]}}$ into $\frac{H}{N}$. Normality of the image of this injection in $\frac{H}{N}$ is clear by construction.

By simplicity of $\frac{H}{N}\cong G$, either $\frac{H{\uhr}_{[k-|T|]}}{N{\uhr}_{[k-|T|]}}\cong \{1_G\}$, in which case we are done immediately, or $\frac{H{\uhr}_{[k-|T|]}}{N{\uhr}_{[k-|T|]}}\cong G$, in which case we are done by \Cref{lem: quotient G}, since $H{\uhr}_{[k-|T|]}=G^{[k-|T|]}$.
\end{proof}

\begin{corollary}\label{cor: support lemma} 
If $[G^k:H]< n^m$, then $H$ has a support of size at most $m$. 
\end{corollary}
\begin{proof}
We prove by induction on $k$. In the base case, $k=1$ and this is true since $H$ has empty support unless $H<G$, in which case $[G:H]\geq n$ by \Cref{lem:min index}. Now assume that the result holds for subgroups of $G^{k}$. Let $H\leq G^{k+1}$, and use \Cref{lem: Goursat} to write $H=\text{Graph}(\frac{G_1}{N_1}\cong\frac{G_2}{N_2})$, where $G_1=\pi_{[k]}(H)$, $N_1=H{\uhr}_{[k]}$, $G_2=\pi_{\{k+1\}}(H)$, and $N_2=H{\uhr}_{\{k+1\}}$.

First consider if $G_2=G$. If $N_2=G$, then we have $H=G_1\times G$, and there exists a support $T$ for $G_1$ of size at most $\log_{n}([G^{k}:G_1])= \log_{n}([G^{k+1}:H])$ that is also a support for $H$. 

Otherwise we have $N_2=\{1_G\}$, since $G$ is simple. Let $T$ be a minimum size support for $G_1$. Since $\frac{G_1}{N_1}\cong \frac{G_2}{N_2}\cong G$, \Cref{cor: quotient G} implies that $N_1$ has a support given by $T\cup \{j\}$ for some $j\in [k]$.

Therefore, $T\cup\{j\}$ is a support for $N_1$ of size at most 
\begin{align*}
    |T\cup \{j\}|&=|T|+1\leq \log_{n}([G^k:G_1])+1 =\log_{n}(\frac{|G^k|}{|G_1|})+1\\
    &= \log_{n}(\frac{|G^{k+1}|}{|G||H|})+1=\log_{n}([G^{k+1}:H])+1-\log_{n}(|G|)\leq \log_{n}([G^{k+1}:H])-1.
\end{align*}
The first inequality follows by induction. The equality $\log_n(\frac{|G^k|}{|G_1|})=\log_n(\frac{|G^{k+1}|}{|G||H|})$ follows from $|H|=|G_1|$, since $H=\text{Graph}(\frac{G_1}{N_1}\cong G)$ and hence $|G||N_1|=|G_1|$. The last inequality uses that $n\leq \sqrt{|G|}$, a fact proved in \Cref{lem:min index}. Now, we claim that $T\cup \{j,k+1\}$ serves as a support for $H$ of size at most $\log_{n}([G^{k+1}:H])$. Let $i\in [k]\setminus( T\cup \{j\})$. Then $H{\uhr}_{\{i\}}\geq N_1|_{\{i\}} = G^{\{i\}}$.

Now we consider when $G_2<G$. In this case, let $T$ be a minimum size support for $N_1$. Then by induction
\begin{align*}
    |T|&\leq \log_{n}([G^k:N_1])=\log_{n}(\frac{|G^k|}{|N_1|})\\
    &= \log_{n}(\frac{|G^{k+1}||G_2|}{|G||H|})= \log_{n}([G^{k+1}:H])-\log_{n}([G:G_2])\leq\log_{n}([G^{k+1}:H])-1.
\end{align*}
The first inequality follows by induction. The equality $|N_1|=\frac{|H|}{|G_2|}$ used in the third step follows because $H$ is the graph of an isomorphism $\frac{G_1}{N_1}\cong \frac{G_2}{N_2}$. This implies that each coset of $N_2$ in $G_2$ yields $|N_1||N_2|$ group elements of $H$. There are $\frac{|G_2|}{|N_2|}$ such cosets, so $|H|=\frac{|N_1||N_2||G_2|}{|N_2|}= |N_1||G_2|$. The last inequality follows from the fact that $n\leq \min_{L<G}[G:K]$, proved in \Cref{lem:min index}. 

Then $T\cup \{k+1\}$ has size at most $\log_{n}([G^{k+1}:H])$. We show that now $T\cup \{k+1\}$ serves as a support for $H$. For any $i\in[k+1]\setminus (T\cup\{k+1\})= [k]\setminus T$, we have $H{\uhr}_{\{i\}}\geq N_1{\uhr}_{\{i\}}= G^{\{i\}}$.
\end{proof}

\subsubsection{The properties of $\mu$}

Now we begin proving that $\mu$ satisfies the hypotheses of \Cref{lem:general lower bound}. \Cref{Q property} is clearly satisfied by definition of $\mu$. We verify that $\mu$ satisfies \Cref{literal property}.

\begin{lemma}\label{lem: literals no $G$-hard chains}
Let $H=\Stab_Q(\chi_i|_\Omega)$ for some $i\in[kn^2]$. Then $\mu(H,K)\leq 1$ for any $K\geq H$.
\end{lemma}
\begin{proof}
In this case $H$ is the $Q$-stabilizer of a function on $\BAR{G}^k$ depending only on $M_{i^*}$ for some $i^*\in[k]$. Therefore, we have that for any $j\in \{2,\dots,2k-1\}\setminus \{2i^*-1,2i^*\}$, we have $H{\uhr}_{\{j\}}=G^{\{j\}}$. As a result, $E(H)\cap E(K)$ is empty and the largest connected component in $([k],E(H)\cap E(K))$ is a single vertex.
\end{proof}

The following lemma proves that $\mu$ satisfies \Cref{intersection property}.
\begin{lemma}\label{lem: simple intersection property}
Let $H\leq K\leq Q$. Let $S\subseteq \{\{i,i+1\}:i\in[k]\}$ be the edge set of a connected component in $([k],E(H)\cap E(K))$. Let $r\in\mathbbm{N}$ and for each $i\in[r]$, fix $H_i\leq L_i\leq Q$ such that $L_i$ is an element of $\mathcal{N}_{Q,\Omega}$. Suppose that $H_1\cap\dots\cap H_r=H$ and $L_1\cap \dots\cap L_r\leq K$. Then there exists $i\in[r]$ such that $S$ is a subset of edges in $([k],E(H_i)\cap E(L_i))$.
\end{lemma}
\begin{proof}
Assume without loss of generality that $S=\{\{j,j+1\}:j\in[k-1]\}$, since in all other cases $S$ is still a contiguous path on $[k]$ and the proof is the same. Notice that $L_i\geq H_i\geq H$ for all $i\in[r]$, so $L_i{\uhr}_{\{2i,2i+1\}}\geq H{\uhr}_{\{2i,2i+1\}}=  \Diag(G^{\{2i,2i+1\}})$ for all $i\in[r]$.

For any $i\in[r]$, assume that there is some $j\in[k-1]$ such that $L_i{\uhr}_{\{2j,2j+1\}}> \Diag(G^{\{2j,2j+1\}})$. By \Cref{cor:diag to full} we have $L_i{\uhr}_{\{2j+1\}}= G^{\{2j+1\}}$. 

Since $L_i$ is an element of $\mathcal{N}_{Q,\Omega}$, there exists a function $f$ on $\Omega$ for which $L_i$ is the $Q$-stabilizer. Let $g_0\in G$ be any group element. For any tuple of matrices $(M_1,\dots,M_k)\in \BAR{G}^k$, 
\begin{align*}
    &(1_G,\dots,1_G,g_0,1_G)f(M_1,\dots,M_k)\\
    =&f(M_1,\dots,M_{k-1},\BAR{g_0}M_k)\\
    =&f(M_1,\dots,M_{k-1}\BAR{g_0}^{-1},M_k)\text{ (since $L_1{\uhr}_{\{2k-2,2k-1\}}\geq \Diag(G^{\{2k-2,2k-1\}})$)}\\
    =&f(M_1,\dots,\BAR{g_1}M_{k-1},M_k) \text{ for some $g_1\in G$}\\
    &\:\:\vdots \\
    =&f(M_1,\dots,\BAR{g_{k-j+1}}M_{j+1},\dots,M_k)\text{ for some $g_{k-j+1}\in G$}\\
    =&f(M_1,\dots,M_k).
\end{align*}
The last equality follows from $f$ being $L_i$ invariant and $L_i{\uhr}_{\{2j+1\}}=G^{\{2j+1\}}$. This shows that $(1_G,\dots,1_G,g_0,1_G)$ is an element of $L_i$, since $L_i$ is the $Q$-stabilizer of $f$.

If for all $i\in[r]$ there exists $j\in[k-1]$ such that $L_i{\uhr}_{\{2j,2j+1\}}> \Diag(G^{\{2j,2j+1\}})$, then we have shown that $(1_G,\dots,1_G,g_0,1_G)$ is an element of $L_1\cap\dots\cap L_r\leq K$ for any $g_0\in G$, a contradiction, since $\{k-1,k\}\in S\subseteq E(K)$ implies that $K{\uhr}_{\{2k-2,2k-1\}}=\Diag(G^{\{2k-2,2k-1\}})$. 

Therefore, there is some $L_i$ such that $L_i{\uhr}_{\{2j,2j+1\}}=\Diag(G^{\{2j,2j+1\}})$ for all $j\in[k-1]$. Now $H_i{\uhr}_{\{2j,2j+1\}}=\Diag(G^{\{2j,2j+1\}})$ for all $j\in[k-1]$ because $H\leq H_i\leq L_i$ implies that for all $j\in[k-1]$,
\begin{align*}
    \Diag(G^{\{2j,2j+1\}})= H{\uhr}_{\{2j,2j+1\}}\leq H_i{\uhr}_{\{2j,2j+1\}}\leq L_i{\uhr}_{\{2j,2j+1\}}=\Diag(G^{\{2j,2j+1\}}).
\end{align*}
Thus, $S$ is a subset of edges in $([k],E(H_i)\cap E(L_i))$.
\end{proof}

It remains to prove that $\mu$ satisfies \Cref{shrinkage property} with $c=n$. A designer attempting to build a small formula would want to decrease $\mu$, and hence the size of the largest connected component in $([k],E(H)\cap E(K))$ quickly, so we would like to characterize the cost of removing edges (which is stated using indices of subgroups). Because a group $H$ with $([k],E(H))$ having some connected component with $\ell\geq 2$ vertices has a shifted diagonal subgroup isomorphic to $G^{\ell}$, and the designer gets to deletes edge when the restrictions $H{\uhr}_{\{2i,2i+1\}}$ are no longer equal to $\Diag(G^{\{2i,2i+1\}})\cong G$, the notion of support (\Cref{def:support}) is useful.

We use \Cref{cor: support lemma} to prove that $\mu$ satisfies \Cref{shrinkage property} with $c=n$.
\begin{lemma}\label{lem: shrinkage by n}
Let $U\leq V\leq Q$ with $V\in\mathcal{N}_{Q,\Omega}$. Let $H\leq L\leq Q$ with $L\in\mathcal{N}_{Q,\Omega}$. Suppose that that $\bigcap_{h\in H}h^{-1}Uh\leq H$ and $\bigcap_{h\in H}h^{-1}Vh\leq L$. Let $[H:U]\leq n^{m-1}$. Then
\begin{align*}
    \mu(U,V)\geq \mu(U\cap H,V)\geq  \frac{\mu(H,L)}{m}.
\end{align*}
\end{lemma}
\begin{proof}
First we prove that $\mu(U,V)\geq \mu(U\cap H,V)$. To do this, consider any edge $\{i,i+1\}\in E(U\cap H)\cap E(V)$. Then $(U\cap H){\uhr}_{\{2i,2i+1\}}=\Diag(G^{\{2i,2i+1\}})$, so that $U{\uhr}_{\{2i,2i+1\}}\geq \Diag(G^{\{2i,2i+1\}})$. 

Assume for sake of contradiction that $U{\uhr}_{\{2i,2i+1\}}> \Diag(G^{\{2i,2i+1\}})$, so that  $U{\uhr}_{\{2i,2i+1\}}= G^{\{2i,2i+1\}}$ by \Cref{cor:diag to full}.

In the first case $H{\uhr}_{\{2i,2i+1\}}\leq \Diag(G^{\{2i,2i+1\}})$. Then, since $(\bigcap_{h\in H}h^{-1}Uh){\uhr}_{\{2i,2i+1\}}=G^{\{2i,2i+1\}}> H{\uhr}_{\{2i,2i+1\}}$, we have $\bigcap_{h\in H}h^{-1}Uh\not\leq H$. This is a contradiction. Otherwise, we have $H{\uhr}_{\{2i,2i+1\}}=G^{\{2i,2i+1\}}$. Then, $(U\cap H){\uhr}_{\{2i,2i+1\}}=G^{\{2i,2i+1\}}$, giving a contradiction with $(U\cap H){\uhr}_{\{2i,2i+1\}}=\Diag(G^{\{2i,2i+1\}})$.

Therefore, it must be that $\{i,i+1\}\in E(U)\cap E(V)$, proving that $E(U\cap H)\cap E(V)\subseteq E(U)\cap E(V)$. This implies that $\mu(U,V)\geq \mu(U\cap H,V)$.

Now we prove the second statement $\mu(U\cap H,V)\geq  \frac{\mu(H,K)}{m}$. Let $\{\{i,i+1\}:i\in R\}$ be the edge set of any connected component of the graph $([k],E(H)\cap E(L))$, where $R\subseteq[k-1]$. Without loss of generality assume that $R=\{1,\dots,r\}$, where $r=|R|$. First, note that
\begin{align}\label{intersection index}
    [H{\uhr}_{\bigcup_{i\in R}\{2i,2i+1\}}:(H\cap U){\uhr}_{\bigcup_{i\in R}\{2i,2i+1\}}]&\leq [H:H\cap U]\leq n^{m-1}.
\end{align}
This follows from the standard group theory fact that if $A\leq B$ and $D$ is any other group, then $[B\cap D:A\cap D]\leq [B:A]$, and that the restriction $\uhr$ is an intersection with a subgroup of $Q$.

Then $\varphi:H{\uhr}_{\bigcup_{i\in R}\{2i,2i+1\}}\to G^r$ is a surjective homomorphism given by
\begin{align*}
    \varphi(g_2,\dots,g_{2r+1})&=(g_2,g_4,\dots,g_{2r})
\end{align*}
Let $T$ be a support for $\varphi((U\cap H){\uhr}_{\bigcup_{i\in R}\{2i,2i+1\}})$ such that $|T|\leq m-1$. Existence of such a $T$ is guaranteed by \Cref{cor: support lemma}, surjectivity of $\varphi$, and (\ref{intersection index}).

For any $i\in R\setminus T$, we have that $(U\cap H){\uhr}_{\{2i,2i+1\}}=\Diag(G^{\{2i,2i+1\}})$. Therefore, for every $i\in R\setminus T$, the edge $\{i,i+1\}$ is present in $E(U\cap H)$.

We claim that $\{i,i+1\}$ is an edge in $E(V)$ as well. Assume otherwise. Then since $V\geq U\geq U\cap H$, we have that $V{\uhr}_{\{2i,2i+1\}}>\Diag(G^{\{2i,2i+1\}})$. By \Cref{cor:diag to full}, $V{\uhr}_{\{2i,2i+1\}}=G^{\{2i,2i+1\}}$. But since $i$ is an element of $R$, we have that $\{i,i+1\}$ is an element of $E(H)\cap E(L)\subseteq E(L)$, so that $L{\uhr}_{\{2i,2i+1\}}=\Diag(G^{\{2i,2i+1\}})$. Now it is impossible to have $\bigcap_{h\in H}h^{-1}Vh\leq L$. We have our contradiction.

This shows that we obtain the graph $([k],E(U\cap H)\cap E(V))$ from the graph $([k],E(H)\cap E(L))$ by removing $|T|\leq m-1$ edges. For any connected component in $([k],E(H)\cap E(L))$ with $t$ vertices, it must be that there exists a connected component in $([k],E(U\cap H)\cap E(V))$ with at least $\frac{t}{m}$ vertices, since $([k],E(H)\cap E(L))$ is a subgraph of a path graph.
\end{proof}

We have now achieved our goal of lower bounding $\beta_{Q,\Omega,d}(G^{k-1},G^{k-1})$.

\begin{theorem}\label{thm:nonabelian simple lower bound}
$\beta_{Q,\Omega,d}(G^{k-1},G^{k-1})\geq n^{d(k^{1/d}-1)}$ and $\beta_{Q,\Omega}(G^{k-1},G^{k-1})\geq n^{\log_2 k}$.
\end{theorem}
\begin{proof}
By definition of the embedding $G^{k-1}<Q$ in (\ref{embedding G}), we have $\mu(G^{k-1},G^{k-1})=k$, since $E(G^{k-1})=\{\{i,i+1\}:{i\in[k-1]}\}$. This section shows that $\mu$ satisfies the hypotheses of \Cref{lem:general lower bound} with $c=n$, so simply apply that result.
\end{proof}

\subsection{$\mathsf{TC}$ lower bounds for cyclic groups of prime order}\label{sec:prime order}
In this section we prove \Cref{thm: simple lower bound} in the case where the group $G$ is abelian using the framework set up in \Cref{sec:framework}. Note that if a finite simple group is abelian, then it must be a cyclic group of prime order, denoted $C_p$ for some prime $p$.

We frequently work with $C_p^k$ as the additive group of the vector space $\FF_p^k$. We fix the following notation. Given a vector $x=(x_1,\dots,x_k)\in \FF_p^k$, define $\text{Supp}(x)=\{i\in[k]:x_i\neq 0\}$, and define the Hamming weight $|x|=|\text{Supp}(x)|$. For $x_1,\dots,x_\ell\in \FF_p^k$, let $\langle x_1,\dots,x_\ell\rangle $ be the subspace generated by $x_1,\dots,x_\ell$.

We want to use \Cref{lem:framework} to prove a lower bound on the formula size of $C_p^{k-1}$-invariant $\mathsf{TC}_d$ formulas computing the function $\Word_{C_p,k}$. As we did in \Cref{sec:nonabelian simple}, we prove a lower bound on $\beta_{Q,\Omega,d}(C_p^{k-1},C_p^{k-1})$ for some choice of $Q$ using \Cref{lem:general lower bound} by constructing a lower-bound witness $\mu$ that satisfies the hypotheses of that theorem.

We choose $Q = C_p^{k}$, viewed as a supergroup $C_p^{k-1} < Q \le S_{kp^2}$, which acts on $\Omega=\BAR {C_p}^k\subseteq \{0,1\}^{kp^2}$ via
\begin{align*}
    (x_1,\dots,x_{k})(M_1,\dots,M_k)&=(\BAR{x_1}M_1,\dots,\BAR{x_{k}}M_k)
\end{align*}
for $(x_1,\dots,x_k)\in C_p^k$.

The embedding $C_p^{k-1}< C_p^{k}=Q$ is given by
\begin{align}\label{embedding C_p}
    (x_1,\dots,x_{k-1})\mapsto (x_1,x_2\dots,x_{k-1},\prod_{i\in[k-1]}x_i^{-1}).
\end{align}
Observe that the coordinate functions $\chi_1,\dots,\chi_m:\{0,1\}^{kp^2}\to\{0,1\}$ are given by the entries of the permutation matrices $\{M_1,\dots,M_k\}$.

Another important point to note is that $\mathcal{B}_{Q,\Omega}$ (and hence also $\mathcal{N}_{Q,\Omega}$) consists of all subgroups of $C_p^k$, since each $H\leq C_p^k$ is the $Q$-stabilizer of the Boolean function on $\Omega=\BAR{C_p}^k$ that evaluates to 1 on $(M_1,\dots,M_k)$ if and only if $(M_1,\dots,M_k)\in H$, where we view each $M_i$ as an element of $C_p$.

We now define the lower-bound witness $\mu$ and prove that it satisfies the hypotheses of \Cref{lem:general lower bound} with $c=p$. We regard subgroups as the additive groups of subspaces of $\FF_p^k$.
\begin{definition}\normalfont
For $H<\FF_p^k$ and any $K\geq H$, define
\begin{align*}
    \mu(H,K)=\max_{W>V:\dim(W)=\dim(V)+1}\min_{v\in V^\perp\setminus W^\perp}|v|.
\end{align*}
Define $\mu(\FF_p^k,\FF_p^k)=0$. The value $\mu(H,K)$ only depends on $H$. Thus, we will suppress notation and write $\mu(H)$ to denote $\mu(H,K)$ for any $K\leq\FF_p^k$.
\end{definition}

\subsubsection{The properties of $\mu$}
First note that \Cref{Q property} is satisfied by $\mu$ by definition. We check that $\mu$ satisfies \Cref{literal property}.
\begin{lemma}\label{lem: literal property finite fields}
Let $H=\Stab_{Q}(\chi_i)$ for some $i\in[kp^2]$. Then $\mu(H)\leq 1$.
\end{lemma}
\begin{proof}
Assume that $H$ is the $Q$-stabilizer of a function $f$ depending on a single $M_{i^*}$ for some $i^*\in[k]$. Since $f$ only depends on $i^*$, we have $H{\uhr}_{[k]\setminus \{i^*\}}=\FF_p^{[k]\setminus\{i^*\}}$. Therefore, any $x\in H^\perp$ has its support $\mathrm{Supp}(x)$ contained in $\{i^*\}$.
\end{proof}

The following lemmas prove that $\mu$ satisfies \Cref{shrinkage property} with $c=p$.
\begin{lemma}\label{lem: conjugate property finite fields}
Let $U,H\leq \FF_p^{k}$ be such that $\bigcap_{h\in H}h^{-1}Uh\leq H$. Then $\mu(U)\geq \mu(H\cap U)$.
\end{lemma}
\begin{proof}
Since $\FF_p^k$ is abelian as an additive group, $U=\bigcap_{h\in H}h^{-1}Uh\leq H$. Therefore, $U=U\cap H$.
\end{proof}

\begin{lemma}[\cite{Rossman2018}, Lemma 3.5]\label{lem: Rossman shrinkage}
Let $H<W\leq \FF_p^k$ with $\dim(W)=\dim(H)+1$. Let $U\leq H$. Then there exists $T> U$ with $\dim(T)=\dim(U)+1$ and
\begin{align*}
    \min_{u\in U^\perp\setminus T^\perp}|u|\geq \frac1{\dim(H)-\dim(U)+1}\min_{h\in H^\perp\setminus W^\perp}|h|.
\end{align*}
Moreover, $H+T=W$.
\end{lemma}

Finally, we prove that $\mu$ satisfies \Cref{intersection property}.
\begin{lemma}
\label{lem: intersection property finite fields}
Let $H\leq \FF_p^{k}$ and assume that $H=H_1\cap\dots\cap H_r$ where each $H_i\leq \FF_p^k$. Then there is some $i\in[r]$ such that $\mu(H_i)\geq\mu(H)$.
\end{lemma}
\begin{proof}
We can assume that $r=2$, since all other cases follow by induction. Let $y\in H^\perp$ be the vector with $|y|=\mu(H)$ and $y=\arg\min_{x\in H^\perp\setminus W^\perp}|x|$, where $\dim(W)=\dim(H)+1$. Let $H=H_1\cap H_2$. This means that $H^\perp={H_1}^\perp+{H_2}^\perp$. 

There exist $y_1\in H_1^\perp$ and $y_2\in {H_2}^\perp$ such that $y=y_1+y_2$. First assume that $y_1$ and $y_2$ are linearly independent. There exists some $a_1, a_2\in \FF_p$ such that $a_1y_1+a_2y_2\in W^\perp$ and one of the $a_i$ is nonzero. Assume otherwise. Then, for any such pair $y_1,y_2$ we would have $W^\perp\cap \langle y_1,y_2\rangle=\{0\}$, and 
\begin{align*}
    \dim(H^{\perp})&=\dim(W^\perp)+\dim(\langle y_1,y_2\rangle)-\dim(W^\perp\cap \langle y_1,y_2\rangle)\\
    &=\dim(W^\perp)+2.
\end{align*}
This contradicts that $\dim(W)=\dim(H)+1$, so assume that $a_1y_1+a_2y_2$ is an element of $W^\perp$. Assume that $a_1\neq 0$. The case where $a_2\neq 0$ follows similarly. Then $y_2$ cannot be an element of $W^\perp$, since otherwise $y=y_1+y_2=\frac1{a_1}(a_1y_1+a_2y_2)-(\frac{a_2}{a_1}+1)y_2$, implying that $y\in W^\perp$. As a result, $\mu(H)\geq |y_2|$. Since $|y|$ is minimal in ${H}^\perp\setminus W^\perp$, we have $|y_2|\geq |y|$. We also have $|y_2|\leq |y|$ since otherwise $\mu(H)\geq |y_2|>|y|$. 

Now consider the space $W+{H_2}$. We have that $H\leq W\cap {H_2}$ so $\dim(W\cap {H_2})\geq \dim(H)$. Thus,
\begin{align*}
    \dim(W+{H_2})&=\dim(W)+\dim({H_2})-\dim(W\cap {H_2})\\
    &\leq \dim(H)+1+\dim({H_2})-\dim(H)\\
    &=\dim({H_2})+1.
\end{align*}
If $\dim(W+{H_2})= \dim({H_2})$, then we would have $W= {H_2}$. Then $y_2$ would be an element of $W^\perp$, a contradiction. Therefore, $\dim(W+{H_2})=\dim({H_2})+1$. Finally, we have that $y_2$ is an element of ${H_2}^\perp\setminus (W+{H_2})^\perp$, and we claim that $|y_2|$ is minimal among all $v\in {H_2}^\perp\setminus(W+{H_2})^\perp$. 

Assume otherwise that there exists another $v\in {H_2}^\perp\setminus(W+{H_2})^\perp$ with $|v|<|y_2|$. Since ${H_2}^\perp\leq{H}^\perp$, $v$ is an element of $H^\perp$. Then it must be that $v$ is an element of $W^\perp$, since otherwise we would contradict minimality of $|y|$ in $H^\perp \setminus W^\perp$. But now $v$ is an element of $W^\perp\cap {H_2}^\perp=v\in (W+{H_2})^\perp$, a contradiction. Therefore, $y_2=\arg\min_{x\in H_2^\perp\setminus (W+H_2)^\perp}|x|$. We have proved in this case that $\mu(H_2)\geq |y_2|\geq |y|=\mu(H)$.

Finally, if $y_1$ and $y_2$ are not linearly independent, then at least one is a nonzero scalar multiple of $y$. Assume it is $y_2$. By the same argument, $y_2=\arg\min_{x\in H_2^\perp\setminus (W+H_2)^\perp}|x|$, and we have $\mu(H_2)\geq |y_2|=|y|= \mu(H)$. 
\end{proof}

Having proved that $\mu$ satisfies all the hypotheses of \Cref{lem:general lower bound}, we have the following result.
\begin{theorem}\label{thm: cyclic prime lower bound}
$\beta_{Q,\Omega,d}(C_p^{k-1},C_p^{k-1})\geq p^{d(k^{1/d}-1)}$ and $\beta_{Q,\Omega}(C_p^{k-1},C_p^{k-1})\geq p^{\log_2 k}$.
\end{theorem}
\begin{proof}
By the embedding (\ref{embedding C_p}), $C_p^{k-1}$ is regarded as the additive group of the subspace $V=\{(x_1,\dots,x_k)\in\FF_p^k:\sum_{i\in[k]}x_i=0\}$. Thus, $\mu(C_p^{k-1},C_p^{k-1})=k$ as witnessed by $\min_{v\in V^\perp\setminus (\FF_p^k)^\perp}|v|=k$. This section shows that $\mu$ satisfies the hypotheses of \Cref{lem:general lower bound} with $c=p$, so we have our result.
\end{proof}

\subsection{Proof of \Cref{thm: simple lower bound}}
Having covered both nonabelian and abelian simple groups, we can now prove our main theorem, which we restate here.
\simpleLB*
\begin{proof}
In the first case $G$ is nonabelian. In this case \Cref{thm:nonabelian simple lower bound} shows that for the choice $Q=\{(g_1,\dots,g_{2k})\in G^{2k}:g-1=g_{2k}=1_G\}$ with the embedding $G^{k-1}\to Q$ given by (\ref{embedding G}), we have $\beta_{Q,\Omega,d}(G^{k-1},G^{k-1})\geq n^{d(k^{1/d}-1)}$ and $\beta_{Q,\Omega}(G^{k-1},G^{k-1})\geq n^{\log_2 k}$.

Any $\Phi\in\mathsf{TC}_d$ such that $\sem{\Phi}=\Word_{G,k}$ must have $\Stab_Q(\sem{\Phi})=\Stab_Q(\Word_{G,k})=G^{k-1}$. If $\Stab_Q(\Phi)\geq G^{k-1}$, then we must have $\Stab_Q(\Phi)= G^{k-1}$, since $\Stab_Q(\Phi)\leq \Stab_Q(\sem{\Phi})$. Therefore, $\mathrm{size}(\Phi)\geq \beta_{Q,\Omega,d}(G^{k-1},G^{k-1})$ by \Cref{lem:framework}. Similarly argue for unbounded-depth formulas, and we have
\begin{align*}
    \mathcal L^{G^{k-1}}_{\mathsf{TC}_d}(\Word_{G,k}) &\ge \beta_{Q,\Omega,d}(G^{k-1},G^{k-1})\geq n^{d(k^{1/d}-1)},\\
    \mathcal L^{G^{k-1}}_{\mathsf{TC}}(\Word_{G,k}) &\ge \beta_{Q,\Omega}(G^{k-1},G^{k-1})\geq n^{\log_2(k)}.
\end{align*}
Otherwise, $G$ is abelian, and is therefore $C_p$ for some prime $p$. Note that then $C_p\to S_p$ is the minimum degree faithful permutation representation. For the choice $Q=C_p^k$ and embedding $C_p^{k-1}\to C_p^k$ given by (\ref{embedding C_p}), $\beta_{Q,\Omega,d}(C_p^{k-1},C_p^{k-1})\geq n^{d(k^{1/d}-1)}=p^{d(k^{1/d}-1)}$ and $\beta_{Q,\Omega}(C_p^{k-1},C_p^{k-1})\geq n^{\log_2 k}=p^{\log_2 (k)}$ by \Cref{thm: cyclic prime lower bound}. 

Any $\Phi\in\mathsf{TC}_d$ such that $\sem{\Phi}=\Word_{C_p,k}$ must have $\Stab_Q(\sem{\Phi})=\Stab_Q(\Word_{C_p,k})=C_p^{k-1}$. If $\Stab_Q(\Phi)\geq C_p^{k-1}$, then we must have $\Stab_Q(\Phi)= C_p^{k-1}$, since $\Stab_Q(\Phi)\leq \Stab_Q(\sem{\Phi})$. Therefore, $\mathrm{size}(\Phi)\geq \beta_{Q,\Omega,d}(C_p^{k-1},C_p^{k-1})$ by \Cref{lem:framework}. Similarly argue for unbounded-depth formulas, and we have
\begin{align*}
    \mathcal L^{C_p^{k-1}}_{\mathsf{TC}_d}(\Word_{C_p,k}) &\geq \beta_{Q,\Omega,d}(C_p^{k-1},C_p^{k-1})\geq p^{d(k^{1/d}-1)},\\
    \mathcal L^{C_p^{k-1}}_{\mathsf{TC}}(\Word_{C_p,k}) &\geq \beta_{Q,\Omega}(C_p^{k-1},C_p^{k-1})\geq p^{\log_2(k)}.
\end{align*}
This completes the proof.
\end{proof}

\section{Proof of \Cref{thm: cyclic prime power LB}}\label{sec:prime power}
In this section we prove the lower bound on the invariant $\mathsf{AC}$ formula size of the word problem on cyclic groups of prime power order stated in \Cref{thm: cyclic prime power LB}. Let $q=p^t$ be some prime power so that $C_q$ is the cyclic group of order $q$. Our technique of using \Cref{lem:general lower bound} for proving our lower bounds for invariant $\mathsf{TC}$ formula size does not carry over to subgroups of $C_q^k$, even though we believe that some variant of them does. Therefore, we revert to the technique used in \cite{Rossman2018}, which relies on the fact that $\mathsf{AC}$ formulas only have gates computing the $\AND$, $\OR$, and $\NOT$ functions. 

We would, however, like to reuse the idea of embedding $C_q^{k-1}$ in a larger group. To do this, we let $C_q^k$ act on inputs $(M_1,\dots,M_k)\in\{0,1\}^{kq^2}$, viewed as circulant $q$-by-$q$ permutation matrices via
\begin{align}\label{embedding C_q}
    (x_1,\dots,x_k)(M_1,\dots,M_k)&=(\BAR{x_1}M_1,\dots,\BAR{x_k}M_k).
\end{align}
Denote this set $\Omega$. The coordinate functions $\chi_i$ then give a single entry in a single $M_i$. 

The shifted diagonal action of $C_q^{k-1}$ on inputs $(M_1,\dots,M_k)\in\Omega$ can now be described by embedding $C_q^{k-1}\to C_q^k$ via
\begin{align*}
    (x_1,\dots,x_{k-1})\to (x_1,\dots,x_{k-1},\prod_{i\in[k-1]}x_i^{-1}).
\end{align*}
Note that the image of this embedding is the subgroup $\{(x_1,\dots,x_k)\in C_q^k: \prod_{i\in[k]}x_i=1_{C_q}\}$.

\subsection{The lower-bound witness}
\begin{definition}\normalfont
Let $C_p\leq C_q$ be the unique cyclic subgroup of order $p$. Let $H\leq C_q^k$. Then $H\cap C_p^k$ can be viewed as a vector space $V_H\leq \FF_p^k$. 
\end{definition}
We will regard $V_H$ as a subgroup of $C_p^k\leq C_q^k$ and also as a subspace of $\FF_p^k$.

\begin{definition}\normalfont
Let $V< W\leq \FF_p^k$ be such that $\dim(W)=\dim(V)+1$. Then define
\begin{align*}
    \mu(V,W)&=\min_{v\in V^\perp\setminus W^\perp}|v|.
\end{align*}
\end{definition}

This $\mu$ plays a role similar to that played by the lower-bound witnesses we defined in our proof of \Cref{thm: simple lower bound} for both the nonabelian and abelian cases.

\begin{lemma}\label{lem: only one perp achieves hard chain}
Let $V< W\leq \FF_p^k$ with $\dim(W)=\dim(V)+1$. If $x\in V^\perp\setminus W^\perp$ is such that $|x|=\min_{v\in V^\perp\setminus W^\perp}|v|$, then there does not exist $y\in V^\perp$, $y\not\in  \langle x\rangle$ such that $\emph{Supp}(y)\subseteq \emph{Supp}(x)$.
\end{lemma}
\begin{proof}
Let $y$ that is not an element of $\langle x \rangle$ be such that $\text{Supp}(y)\subseteq \text{Supp}(x)$ and $y$ is an element of $V^\perp$. First we argue that there are $a,b\in \FF_p$ such that $ax+by$ is in $W^\perp$ and at least one of $a,b$ is nonzero. Assume otherwise. Then, we would have $U^\perp\cap \langle x,y\rangle=\{0\}$, and
\begin{align*}
    \dim(V^\perp)&=\dim(W^\perp)+\dim(\langle x,y\rangle)-\dim(U^\perp\cap \langle x,y\rangle)\\
    &=\dim(W^\perp)+2.
\end{align*}
This contradicts $\dim(U)=\dim(V)+1$, so for some $a,b\in \FF_p$, we have that $ax+by$ is an element of $U^\perp$. Assume that $a=1$, so that $x+by$ is in $U^\perp$.

If $y$ is an element of $W^\perp$ then $x=\frac1a(ax+by-by)\in W^\perp$, a contradiction. Therefore we must have that $y$ is not in $W^\perp$. If $\text{Supp}(y)\subsetneq \text{Supp}(x)$ then we contradict minimality of $|x|$ in $V^\perp\setminus W^\perp$, so assume that $\text{Supp}(y)=\text{Supp}(x)$. We can also assume that $|\text{Supp}(x)|\geq 2$, since otherwise $y$ is a scalar multiple of $x$. Finally, we also assume that $p\neq 2$, since in $\FF_2^k$, the set $\text{Supp}(x)$ completely determines $x$.

Since $y$ is not in $W^\perp$, we must have that $x+b'y$ in $W^\perp$ implies that $b'=b$. However, there is at least one other $b'\in \FF_p$ with $b'\neq b$ such that $|x+b'y|< |x|$, since $p\neq 2$ and $y$ and $x$ are not scalar multiples of each other but have the same support. This contradicts minimality of $|x|$ in $V^\perp\setminus W^\perp$, since $x+b'y$ is not in $W^\perp$.
\end{proof}

\begin{corollary}\label{cor: projection is codimension 1}
Let $V<W\leq \FF_p^k$ with $\dim(W)=\dim(V)+1$. Let $v=\arg\min_{v\in V^\perp\setminus W^\perp}|v|$. Then $\pi_{\emph{Supp}(v)}(V)=\langle \pi_{\emph{Supp}(v)}(v)\rangle ^\perp$.
\end{corollary}

For notational convenience, regard $C_q$ as the additive group of the ring $\ZZ/q\ZZ$, the integers modulo $q$. For $a\in\ZZ/q\ZZ$, let $(a)$ be the ideal generated by $a$. Then $C_p\leq C_q$ is the additive subgroup given by the ideal $(p^{t-1})$.
\begin{lemma}\label{lem: literal property AC}
Let $\chi_i:\Omega\to \{0,1\}$ be a coordinate function. If $f$ is invariant under $V\leq C_p^k$ and not constant on $W>V$ with $\dim(W)=\dim(V)+1$, then $\mu(V,W)\leq 1$, viewing $\Omega$ as the group $(\ZZ/q\ZZ)^k$.
\end{lemma}
\begin{proof}
Assume for sake of contradiction that $\mu(V,W)\geq 2$. Let $v=\arg\min_{v\in V^\perp\setminus W^\perp}|v|$. Assume that $\{1,2\}\subseteq \text{Supp}(v)$. Then since $\chi_i$ is not constant on $W$, we have that there exist $p^{t-1}(x_1,\dots,x_k)\in W$ such that $\chi_i((0,\dots,0))\neq f(p^{t-1}(x_1,x_2,\dots,x_k))$ and $x_1\neq 0$ by \Cref{cor: projection is codimension 1}.

Again, by \Cref{cor: projection is codimension 1}, we have that there exists $p^{t-1}(x_1',x_2',\dots,x_k')\in V$ such that $x_2'\neq x_2$. Then, by $V$-invariance, we have $\chi_i(p^{t-1}(x_1,x_2,\dots,x_k)=\chi_i(p^{t-1}(x_1-x_1',x_2-x_2',\dots,x_k-x_k')\neq \chi_i((0,\dots,0))$. Since $x_2+x_2'\neq 0$, we have that $\chi_i$ depends on both $M_1$ and $M_2$, a contradiction.
\end{proof}

\begin{lemma}\label{lem: shrinkage by q}
Let $H\leq C_q^k$. If $|H|\geq q^m$, then $\dim(V_H)\geq m$.
\end{lemma}
\begin{proof}
For each $i\in \{0,\dots,t-1\}$, let $V_i=\frac{H\cap (p^{i})^k}{H\cap (p^{i+1})^k}$. Now notice that $V_i\leq V_{i+1}$ for all $i\in \{0,\dots,t-2\}$, since if $(x_1,\dots,x_k)+(p^t)^k$ is an element of $\frac{H\cap (p^{i-1})^k}{H\cap (p^{i})^k}$, then $p(x_1,\dots,x_k)$ is an element of $\frac{H\cap (p^{i})^k}{H\cap (p^{i+1})^k}$. Moreover,
\begin{align*}
    \prod_{i=1}^{t-1}|V_i|&=\frac{|H|}{|H\cap (p)^k|}\cdot \frac{|H\cap (p)^k|}{|H\cap (p^2)^k|}\cdots \frac{|H\cap (p^{t-1})^k|}{|H\cap (0)^k|}
    =\frac{|H|}{|H\cap(0)^k|}
    =|H|.
\end{align*}
If $|V_{t-1}|=|V_H|<p^m$, then $|V_i|<p^m$ for all $i\in\{0,\dots,t-1\}$. Then we have 
\begin{align*}
    |H|=\prod_{i=0}^{t-1}|V_i|<p^{tm}=q^m.
\end{align*}
This proves the contrapositive.
\end{proof}
In the following proof, we will identify subspaces of $\FF_p^k$ with the corresponding subgroups of $C_q^k$. That is, for $W\leq\FF_p^k$, we also say that $W$ is the subgroup $L\cap C_p^k\leq C_q^k$, where $L$ is any subgroup of $C_q^k$ such that $V_L=W$. We will also say that a formula $\Phi$ is nonconstant on a subset of $\Omega$ (viewed as $C_q^k$) if and only if $\sem{\Phi}$ is nonconstant on that subset.
\begin{theorem}\label{thm: cyclic prime power general mu}
Let $H\leq C_q^k$. Let $W>V_H$ be such that $\dim(W)=\dim(V_H)+1$. Let $\Phi\in\mathsf{AC}_d$ be a formula taking inputs in $\Omega$ that is $H$-invariant but nonconstant on inputs in $W$. Then for $d\geq 1$,
\begin{align*}
    \emph{size}(\Phi)\geq \frac{|H|q^{d(\mu(V_H,W)^{1/d}-1)}}{q^{\dim(V_H)}}.
\end{align*}
Moreover, as $d$ approaches $\infty$, the lower bound goes to
\begin{align*}
    \emph{size}(\Phi)\geq \frac{|H|q^{\log_2(\mu(V_H,W))}}{q^{\dim(V_H)}}.
\end{align*}
\end{theorem}
\begin{proof}
We prove by induction on $d$ and begin at $d=0$. Interpret for $d=0$
\begin{align*}
    q^{d(\mu(V_H,W)^{1/d}-1)}&=\begin{cases}
    1 &\text{if } \mu(V_H,W)=1,\\ 
    \infty &\text{otherwise.}
  \end{cases}
\end{align*}
In both cases we are done by \Cref{lem: literal property AC} and \Cref{lem: shrinkage by q}. Now assume that the result holds for all $\Psi\in\mathsf{AC}_d$. Note that the base case fits into our induction, since these are simply the limits taken as $d\to 0$.

Let $\Phi=(\OR,I)$ be a $\mathsf{AC}_{d+1}$ formula that is $H$-invariant but nonconstant under inputs in $W>V_H$ with $\dim(W)=\dim(V_H)+1$. We can assume that the top gate of $\Phi$ is $\OR$, since all the $\NOT$ gates can be pushed to the bottom of the formula and the AND case follows symmetrically. Letting $H$ act on $I$, we can break into $H$-orbits, and write $I=I_1\cup\dots\cup I_r$. Let $\Phi_i$ be the formula $(\OR,I_i)$. We claim that there is some $\Phi_i$ that is not $W$-invariant. 

Since $\Phi=\bigvee_{i\in[r]}\Phi_i$, we cannot have $\Phi_i$ identically 1 on $W$ for any $i$. Otherwise, we would have $\Phi$ identically 1 on $W$. If $\Phi_i$ is identically 0 on $W$ for all $i$, then so is $\Phi$. Therefore, there is some $\Phi_i$ that is nonconstant on inputs in $W$. Since $\text{size}(\Phi_i)\geq\text{size}(\Phi)$, we can assume that $\Phi_i$ is all of $\Phi$ and that the action of $H$ is transitive on $I$.

Let $\Psi\in I$. Let $U=\Stab_H(\Psi)$. Notice that $U$ is also the $H$-stabilizer for all $\Psi^h$ with $h\in H$ by the orbit stabilizer theorem and $C_q^k$ being abelian.

Let $m=\dim(V_H)-\dim(V_U)+1$. By \Cref{lem: Rossman shrinkage}, there exists $T>V_U$ with $\dim(T)=\dim(V_U)+1$ such that $\mu(V_U,T)\geq \frac{\mu(V_H,W)}m$ and $V_H+T=W$. 

We claim that there is $h\in H$ such that $\Psi^h$ is nonconstant on $T$. We break into cases.

In the first case, $\Phi$ is identically 0 on $V_H$ and identically 1 on some coset $w+V_H$ of $V_H$ in $W$. Then some $\Psi$ is identically 0 on $V_H$ and not identically 0 on $w+V_H$. Let $w'\in w+V_H$ be such that $\Psi(w')=1$. Let $t\in T$ and $h\in V_H$ be such that $h+t=w'$. Such a pair exists because $T+V_H=W$. Then, $\Psi^h(t)=\Psi(h+t)=\Psi(w')=1$. We have $\Psi^h$ is identically 0 on $U$ and nonconstant on $T$, as $\Psi^h(t)=1$ and $\Psi^h(0)=\Psi(h)=0$.

In the second case, $\Phi$ is identically 1 on $V_H$ and identically 0 on some coset $w+V_H$ of $V_H$ in $W$. Then some $\Psi$ is identically 0 on $w+V_H$ and not identically 0 on $V_H$. Let $h\in V_H$ be such that $\Psi(h)=1$. Let $t\in T$ be such that $t+h$ is in the coset $w+V_H$. Such $t$ exists by the following argument. Let $h'\in V_H$ and $t\in T$ be such that $h'+t=w$. Then $h'=h''+h$ where $h''$ is in $V_H$. Therefore
\begin{align*}
    h+h''+t&=w  \iff
    h+t=w-h'' \implies
    h+t\text{ is in } w+V_H.
\end{align*}
Then $\Psi^h(t)=\Psi(h+t)=\Psi(w)=0$. Again, $\Psi^h$ is nonconstant on $T$ because $\Psi^h(0)=\Psi(h)=1$ and $\Psi^h(t)=\Psi(w)=0$.

As $\Psi^h$ has depth $d$, is $U$-invariant, and is not constant on $T>V_U$ where $\dim(T)=\dim(V_U)+1$, we have by induction that $\text{size}(\Psi^h)\geq \frac{|U|q^{(d-1)(\mu(V_U,T)^{1/(d-1)}-1)}}{q^{\dim(V_U)}}$. By the orbit-stabilizer theorem, $|I|=[H:U]$, so we have
\begin{align*}
    \text{size}(\Phi)&\geq [H:U]\text{size}(\Psi^h)\\
    &\geq \frac{q^{\dim(V_H)}|H|}{q^{\dim(V_H)}|U|}\cdot\frac{|U|q^{(d-1)(\mu(V_U,T)^{1/(d-1)}-1)}}{q^{\dim(V_U)}}\\
    &\geq \frac{|H|q^{(d-1)(\mu(V_U,T)^{1/(d-1)}-1)}}{q^{\dim(V_H)}}\cdot \frac{q^{\dim(V_H)}}{q^{\dim(V_U)}}\\
    &\geq \frac{|H|q^{m-1+(d-1)((\frac{\mu(V_H,T)}m)^{1/(d-1)}-1)}}{q^{\dim(V_H)}}\\
    &\geq \frac{|H|q^{d(\mu(V_H,W)^{1/d}-1)}}{q^{\dim(V_H)}}.
\end{align*}
The last inequality follows from a basic optimization over $m$ using elementary calculus. 

The second claim about the limit as $d\to\infty$ again follows from integrality of $\mu$.
\end{proof}

We now obtain our lower bound for $\mathsf{AC}$ formulas computing $\Word_{C_q,k}$ as an immediate consequence of this result.
\CyclicACLB*
\begin{proof}
The embedding (\ref{embedding C_q}) lets us view $C_q^{k-1}\leq C_q^k$ as the subgroup $\{(x_1,\dots,x_k)\in (\ZZ/q\ZZ)^k:\sum_{i\in[k]}x_i=0\}$.

Then, $\mu(V_{C_q^{k-1}},\FF_p^k)=k$. The word problem is nonconstant on $C_p^k\leq C_q^k$. Moreover, $|C_q^{k-1}|=q^{k-1}=q^{\dim(V_{C_q^{k-1}})}$. Finally, we note that any $\Phi\in\mathsf{AC}_d$ computing $\Word_{C_q,k}$ must be nonconstant on $C_p^k\leq C_q^k$, so satisfies the hypotheses of \Cref{thm: cyclic prime power general mu}. Thus, that theorem gives the result.
\end{proof}

\section{Groups with small cyclic $p$-subgroups}\label{sec:large gaps}
In this section we show that given some conjectures on the Mersenne numbers, the $\mathsf{AC}$ formula size lower bounds implied by \Cref{thm: simple lower bound} may be much stronger than those implied by \Cref{thm: cyclic prime power LB}. This justifies the effort put into \Cref{thm: simple lower bound}, especially in the nonabelian case.

In some cases for $G$, the two bounds may not be very far apart. For instance, when $G=A_n$, the alternating group on $n$ elements when $n\geq 5$, then \Cref{thm: simple lower bound} implies $\mathcal{L}_{\mathsf{AC}_d}^{A_n^{k-1}}(\Word_{A_n,k})\geq n^{d(k^{1/d}-1)}$. Since there is a prime $p\geq \frac{n}{2}$ by Bertrand's postulate, and therefore a subgroup of $A_n$ isomorphic to $C_p$, even \Cref{thm: simple lower bound} applied to just abelian groups gives a bound $\mathcal{L}_{\mathsf{AC}_d}^{A_n^{k-1}}(\Word_{A_n,k})\geq (\frac{n}{2})^{d(k^{1/d}-1)}$. These two bounds have similar bases ($\frac{n}{2}$ vs. $n$) for the exponent.

We give conditional examples of nonabelian simple $G$ where applying \Cref{thm: simple lower bound} or \Cref{thm: cyclic prime power LB} to a cyclic subgroup of prime power of $G$ does not yield a lower bound on $\mathcal{L}_{\mathsf{AC}_d}^{G^{k-1}}(\Word_{G,k})$ that is near the lower bound implied by \Cref{thm: simple lower bound} applied to $G$ itself. Our examples are $n$-dimensional projective special linear groups over $\FF_2$, denoted $\mathrm{PSL}(n,\FF_2)$. These groups are nonabelian simple. The orders of the cyclic subgroups of prime power order of these groups are related to numbers of the form $2^i-1$.

\begin{definition}\label{def:mersenne}\normalfont
The \emph{Mersenne numbers} are the numbers of the form $2^i-1$ for $i\in\mathbbm{N}$. The \emph{Mersenne primes} are the Mersenne numbers that are prime.
\end{definition}
The two conjectures on which our examples of large gaps between the lower bounds implied by the two theorems depend concern the properties of the Mersenne numbers and Mersenne primes. First, the Lenstra-Pomerance-Wagstaff Conjecture \cite{wagstaff1983divisors} implies that for any $x\in\mathbbm{N}$, the expected number of $x\leq i\leq 2x$ such that $2^i-1$ is prime is approximately 2. That is, $\mathbbm{E}_{x\in\mathbbm{N}}[|\{i:x\leq i\leq 2x, 2^i-1 \text{ prime}\}|]\approx 2$.

Also, it has been conjectured that all Mersenne numbers are square-free \cite{warren1967square}. We call this the Square-Free Conjecture.

We show that, given these two conjectures, there are infinitely many $n$ such that $\mathrm{PSL}(n,\FF_2)$ has no cyclic subgroup of prime power order with order greater than $2^{5n/6}$. 

The claim implies that the best bound that can be obtained by \Cref{thm: cyclic prime power LB} is given by $\mathcal{L}_{\mathsf{AC}_d}^{\mathrm{PSL}(n,\FF_2)^{k-1}}(\Word_{\mathrm{PSL}(n,\FF_2),k})\geq 2^{(5n/6) d(k^{1/d}-1)}$. On the other hand, the lower bound obtained by \Cref{thm: simple lower bound} is $\mathcal{L}_{\mathsf{AC}_d}^{\mathrm{PSL}(n,\FF_2)^{k-1}}(\Word_{\mathrm{PSL}(n,\FF_2),k})\geq 2^{(n-1)  d(k^{1/d}-1)}$, since $\mathrm{PSL}(n,\FF_2)$ is simple for $n\geq 3$ and has minimal degree permutation representation of degree $2^{n}-1$ \cite{vasilyev1996minimal}, obtained by its permutation action on $n$-dimensional projective space. The gap between the base in the exponents ($2^{5n/6}$ vs. $2^{n-1}$) is quite significant.

We first show that conditional on the Square-Free Conjecture, the largest cyclic subgroup of prime power order $\mathrm{PSL}(n,\FF_2)$ for $n$ large enough has order equal to a Mersenne prime. It suffices to do so for $\mathrm{GL}(n,\FF_2)$, since $\mathrm{PSL}(n,\FF_2)$ is a quotient of a subgroup of $\mathrm{GL}(n,\FF_2)$. We follow a proof by van Leeuwen.

\begin{lemma}[\cite{294524}]\label{lem:GL exponent}
Let $n\geq 2$. Suppose that $g\in\mathrm{GL}(n,p)$ has maximum order among all elements of prime power order. If the Square-Free Conjecture holds, then the order of $g$ is $p$ for some Mersenne prime $p<2^n$.
\end{lemma}
\begin{proof}
Fix some prime $P$. Let $g\in\mathrm{GL}(n,p)$ and let $f$ be its minimal polynomial. Note that $f$ is the product of irreducible polynomials. 

The order of $g$ is then the minimal $e$ such that $X^e-1$ is divisible by $f$. Let $d$ be the degree of $f$. The polynomial $X^{p^d-1}-1$ is the product of all monic irreducible polynomials with degree dividing $d$ except $X$.

If $f$ had no repeated irreducible factors then $f$ would divide $X^{p^d-1}$. We must account for the powers of irreducible factors, which arise when $g$ has nontrivial Jordan blocks. Consider the case in which $X^{p^d-1}$ is divisible by $f$, an irreducible, but not by $f^2$.

In this case, $X^{q^d-1}-1$ is divisble by $f$ but not by $f^2$. Then $X^{q^d-1}=1+af$ with $a$ not divisible by $f$. The binomial formula implies that $X^{(q^d-1)i}\not\equiv 1\text{(mod }f^2)$ for $i\in[p-1]$. On the other hand, the Frobenius endomorphism on $\FF_p[X]$ implies that $X^{(q^d-1)p}=1+a^pf^p$. Continuing, we find that $X^{(q^d-1)p^{\ceil{\log_p(i)}}}\equiv 1 \text{(mod }f^i)$.

We find that the order of $g$ with minimal polynomial $f$ divides $(q^d-1) p^{\ceil{\log_p(i)}}$ for some integer $i$ such that $ \deg(f)i\leq  n$.

Apply this to the case $p=2$. If $g$ has prime power order, either its order is divisible by 2, in which case the order of $g$ is at most $\ceil{\log_2(n)}$. Otherwise the order of $g$ is divisble by an odd prime, and it divides $2^{d}-1$ for some $d$. By the Square-Free Conjecture, $2^d-1$ is square-free, so the order of $g$ must be a prime dividing $2^d-1$.

By the relation $\mathrm{GCD}(2^k-1,2^\ell-1)=2^{\mathrm{GCD}(k,\ell)}-1$, we have that the order of $g$ is a Mersenne prime less than $2^n$. 

Since $n\geq 2$, to maximize order over all elements of prime power order, we would choose an element $g$ such that its order is odd and is therefore a Mersenne prime.
\end{proof}

\begin{corollary}\label{cor:small p-subgroups PSL}
Given the Lenstra-Pomerance-Wagstaff Conjecture and Square-Free Conjecture, there exist infinitely many $n\in\mathbbm{N}$ such that $\mathrm{PSL}(n,\FF_2)$ has no cyclic subgroup of prime power order greater than $2^{5n/6}$.
\end{corollary}
\begin{proof}
Given the Lenstra-Pomerance-Wagstaff Conjecture, let $x\in\mathbbm{N}$ be such that $|\{i:x\leq i\leq 2x, 2^i-1 \text{ prime}\}|\leq 2$; there are infinitely many such $x$. By averaging, there exists $x\leq x_1<x_2\leq 2x$ such that there are no Mersenne primes of the form $2^i-1$ with $x_1\leq i\leq x_2$ and $x_2-x_1\geq \frac{x}6$. Set $n=x_2$. We can assume that $n\geq 2$.

By \Cref{lem:GL exponent}, given the Square-Free Conjecture, the maximum cyclic subgroup of prime power order in $\mathrm{PSL}(n,\FF_2)$ has order $p$, where $p$ is a Mersenne prime less than $2^n$. There are no Mersenne primes between $2^{x_1}-1$ and $2^n-1$, so $p\leq 2^{x_1}-1\leq 2^{5n/6}$, completing the proof.
\end{proof}

\section{Future directions}
In this paper we have proved tight bounds on the size of invariant $\mathsf{TC}$ formulas for the word problem on all finite simple groups (\Cref{thm: simple lower bound}). We also prove a tight lower bound on the size of invariant $\mathsf{AC}$ formulas for the word problem of cyclic groups of prime power order (\Cref{thm: cyclic prime power LB}). One direction for future work would be to improve this $\mathsf{AC}$ lower bound to a $\mathsf{TC}$ lower bound.
\begin{question} 
Do we have $\mathcal{L}^{C_q^{k-1}}_{\mathsf{TC}_d}(\Word_{C_q,k})\geq q^{d(k^{1/d}-1)}$?
\end{question}
Another direction for progress is characterizing the invariant $\mathsf{TC}$ or $\mathsf{AC}$ formula size for all finite groups. With \Cref{cor:all finite groups} we have answered the following question is true for simple groups and abelian groups.
\begin{question}
What is $\mathcal{L}_{\mathsf{AC}_d}^{G^{k-1}}(\Word_{G,k})$ for a finite group $G$?
\end{question}

Finally, it is natural to ask whether these invariant formulas are optimal even without the restriction of invariance. Currently \Cref{cor:upper} gives the best known upper bounds on the formula size for $\Word_{S_n,k}$.
\begin{question}
Do we have $\mathcal{L}_{\mathsf{TC}_d}(\Word_{S_n,k})\geq n^{\Omega(d(k^{1/d}-1))}$?
\end{question}
An affirmative answer to this question, even for $k=\log(n)$, would prove that $\mathsf{Logspace}\neq \mathsf{NC}^1$.

\bibliographystyle{plain}
\bibliography{Symmetric_formulas_refs}

\end{document}